\documentclass[11pt]{article}

\usepackage{natbib}
\usepackage{fullpage}

\usepackage{amsmath, amsthm, amssymb}
\usepackage{verbatim}
\usepackage{xcolor}
\usepackage{enumitem}
\setlist[itemize]{leftmargin=*}

\usepackage[ruled,linesnumbered]{algorithm2e} % For algorithms

\usepackage{graphicx}

\newtheorem{theorem}{Theorem}
\newtheorem{lemma}{Lemma}
\newtheorem{corollary}{Corollary}

\theoremstyle{definition}
\newtheorem{definition}{Definition}
\newtheorem{example}{Example}

\renewcommand{\paragraph}[1]{\smallskip\noindent\textbf{#1}}

\title{Automated Mechanism Design for Classification with Partial Verification}

\author{
    Hanrui Zhang \\
    Duke University \\
    \texttt{hrzhang@cs.duke.edu}
\and
    Yu Cheng \\
    University of Illinois at Chicago \\
    \texttt{yucheng2@uic.edu}
\and
    Vincent Conitzer \\
    Duke University \\
    \texttt{conitzer@cs.duke.edu}
}

\begin{document}

\maketitle

\begin{abstract}
    We study the problem of automated mechanism design with partial verification, where each type can (mis)report only a restricted set of types (rather than any other type), induced by the principal's limited verification power.
    We prove hardness results when the revelation principle does not necessarily hold, as well as when types have even minimally different preferences.
    In light of these hardness results, we focus on truthful mechanisms in the setting where all types share the same preference over outcomes, which is motivated by applications in, e.g., strategic classification.
    We present a number of algorithmic and structural results, including an efficient algorithm for finding optimal deterministic truthful mechanisms, which also implies a faster algorithm for finding optimal randomized truthful mechanisms via a characterization based on convexity.
    We then consider a more general setting, where the principal's cost is a function of the combination of outcomes assigned to each type.
    In particular, we focus on the case where the cost function is submodular, and give generalizations of essentially all our results in the classical setting where the cost function is additive.
    Our results provide a relatively complete picture for automated mechanism design with partial verification.
\end{abstract}

\section{Introduction}

Agents are often {\em classified} into a variety of categories, some more desirable than others.  Loan applicants might be classified in various categories of risk, determining the interest they would have to pay.  University applicants may be classified into categories such as ``rejected,'' ``wait list,'' ``regular accept,'' and ``accept with honors scholarship.''  Meanwhile universities might themselves be classified into categories such as ``most competitive,'' ``highly competitive,'' etc.
In line with the language of {\em mechanism design} (often considered part of {\em game theory}), we assume that each agent (i.e., the entity being classified) has a {\em type}, corresponding to the agent's true financial situation, ability as a student, or competitiveness as a university.
This type is information that is private to the agent.
In most applications of mechanism design, the type encodes the agent's {\em preferences}.  For example, in an auction, an agent's type is how much he values the outcome where he wins the auction.
In contrast, in our setting, the type does not encode the agent's preferences: in the examples above, typically any agent has the same preferences over outcomes, regardless of the agent's true type.  Instead, the type is relevant to the objective function of the {\em principal} (the entity doing the classification), who wants to classify the agents into a class that fits their type.

Often, in mechanism design, it is assumed that an
agent of any type can report any other type (e.g., bid any value in an auction), and outcomes are based on these reports.
Under this assumption, our problem would be hopeless: 
every agent would always simply report whatever type gives the most favorable outcome, so we could not at all distinguish agents based on their true type.  But 
%this is assuming that agents are not limited in what they can misreport.  While that is the standard assumption in mechanism design and sensible in contexts such as auctions, here it is 
in our context this assumption is 
not sensible: while an agent may be able to take some actions that affect how its financial situation appears, it will generally not be possible for a person in significant debt and without a job to successfully imitate a wealthy person with a secure career.
This brings us into the less commonly studied domain of {\em mechanism design with partial verification}~\citep{green1986partially,yu2011mechanism}, in which not every type can misreport every other type.  That is, each type has certain other types that it can misreport.
A standard example in this literature is that it is possible to have arrived later than one really did, but not possible to have arrived earlier.
(In that case, the arrival time is the type.)
In this paper, however, we are interested in more complex misreporting (in)abilities.

What determines which types can misreport (i.e., successfully imitate) which other types?  This is generally specific to the setting at hand.
\citet{zhang2019samples} consider settings in which different types produce ``samples'' (e.g., timely payments, grades, admissions rates, ...) according to different distributions. They characterize which types can distinguish themselves from which other types in the long run, in a model in which agents can either (1) manipulate these samples before they are submitted to the principal, by either withholding transforming some of them in limited ways, or (2) choose the number of costly samples to generate~\citep{zhang2019samples,zhang2019distinguishing,zhang2021classification}.  In this paper, we will take as given which types can misreport which other types; this relation may result from applying the above characterization result, or from some other model.

Our goal is: given the misreporting relation, agents' preferences, and the principal's objective, can we efficiently compute the optimal (single-agent) mechanism/classifier, which assigns each report to an outcome/class?  This is a problem in {\em automated mechanism design}~\citep{conitzer2002complexity,conitzer2004self}, where the goal is to compute the optimal mechanism for the specific setting (outcome space, utility and objective functions, type distribution, ...) at hand.  Quite a bit is already known about the complexity of the automated mechanism design problem, and with partial verification, the problem is known to become even harder~\citep{auletta2011alternatives,yu2011mechanism,kephart2015complexity,kephart2016revelation}.  The structural advantage that we have here is that, unlike that earlier work, we are considering settings where all types have the same preferences over outcomes.  This allows us positive results that would otherwise not be available.

\subsection{Our Results and Techniques}

Throughout the paper, we assume agents have {\em utility} functions which they seek to maximize, and the principal has a {\em cost} function which she seeks to minimize.

\paragraph{General vs.\ truthful mechanisms.}
We first set out to investigate the problem of automated mechanism design with partial verification in the most general sense, where there is no restriction on each type's utility function. %, and (unlike in the standard case) the revelation principle does not necessarily hold.
In light of previously known hardness results, although the most general problem is unlikely to be efficiently solvable, one may still hope to identify maximally nontrivial special cases for which efficient algorithms exist.
In order to determine the boundary of tractability, our first finding, Theorem~\ref{thm:np-hardness-without-revelation}, shows that when the revelation principle does not hold, % no efficient algorithm exists even for the minimally nontrivial setting (unless $\mathsf{P} = \mathsf{NP}$).
it is $\mathsf{NP}$-hard to find an optimal (randomized or deterministic) mechanism even if (1) there are only $2$ outcomes and (2) all types share the same utility function.\footnote{The {\em revelation principle} states that if certain conditions hold on the reporting structure, then it is without loss of generality to focus on {\em truthful} mechanisms, in which agents are always best off revealing their true type.
We will discuss below a necessary and sufficient condition for the revelation principle to hold in our setting.}
In other words, without the revelation principle, no efficient algorithm exists even for the minimally nontrivial setting.
We therefore focus our attention on cases where the revelation principle holds, or, put in another way, on finding optimal  truthful mechanisms.

\paragraph{General vs.\ structured utility functions.}
The above result, as well as prior results on mechanism design with partial verification~\citep{auletta2011alternatives,yu2011mechanism,kephart2015complexity,kephart2016revelation}, paints a clear picture of intractability when the revelation principle does not hold.
But prior work also often suggests that this is indeed the boundary of tractability.
This is in fact true if we consider optimal randomized truthful mechanisms, which can be found by solving a linear program with polynomially many variables and constraints if the number of agents is constant \citep{conitzer2002complexity}.
However, as our second finding (Theorem~\ref{thm:np-hardness-with-general-utility}) shows, the case of {\em deterministic} mechanisms is totally different --- even with $3$ outcomes and single-peaked preferences over outcomes, it is still $\mathsf{NP}$-hard to find an optimal deterministic truthful mechanism (significantly improving over earlier hardness results for deterministic mechanisms~\citep{conitzer2002complexity,conitzer2004self}).
In other words, optimal deterministic truthful mechanisms are almost always hard to find whenever types have different preferences over outcomes.
This leads us to what appears to be the only nontrivial case left, i.e., where all types share the same preference over outcomes.
But this case is important: as discussed above, it in fact nicely captures a number of real-world scenarios of practical importance, and will be the focus in the rest of our results.

\paragraph{Efficient algorithm for deterministic mechanisms.}
Our first algorithmic result (Theorem~\ref{thm:additive-deterministic}) is an efficient algorithm for finding optimal deterministic truthful mechanisms with identical preferences in the presence of partial verification.
The algorithm works by building a directed capacitated graph, where each deterministic truthful mechanism corresponds bijectively to a finite-capacity $s$-$t$ cut.
The algorithm then finds an $s$-$t$ min-cut in polynomial time, which corresponds to a deterministic truthful mechanism with the minimum cost.

\paragraph{Condition for deterministic optimality and faster algorithm for randomized mechanisms.}
We then consider randomized mechanisms.
We aim to answer the following two natural questions.
\begin{itemize}
    \item In which cases is there a gap between optimal deterministic and randomized mechanisms, and how large can this gap be?
    \item While LP formulations exist for optimal randomized truthful mechanisms in general, is it possible to design theoretically and/or practically faster algorithms when types share the same utility function?
\end{itemize}
The answers to these questions turn out to be closely related.

For the first question, we show that the gap in general can be arbitrarily large (Example~\ref{ex:gap}).
On the other hand, there always exists an optimal truthful mechanism that is deterministic whenever the principal's cost function is convex with respect to the common utility function (Lemma~\ref{lem:additive-convex}).
In order to prove this, we show that without loss of generality, an optimal truthful mechanism randomizes only between two consecutive outcomes (when sorted by utility) for each type, and present a way to round any such mechanism into a deterministic truthful mechanism, preserving the cost in expectation.

For the second question, we give a positive answer, by observing that with randomization, essentially only the convex envelope of the principal's cost function matters.
This implies a reduction from finding optimal randomized mechanisms with general costs, to finding optimal randomized mechanisms with convex costs, and -- via our answer to the first question (Lemma~\ref{lem:additive-convex}) -- to finding optimal deterministic mechanisms with convex costs.
As a result, finding optimal randomized truthful mechanisms is never harder than finding optimal deterministic truthful mechanisms with convex costs.
Combined with our algorithm for the latter problem (Theorem~\ref{thm:additive-deterministic}), this reduction implies a theoretically and practically faster algorithm for finding optimal randomized truthful mechanisms when types share the same utility function.

\paragraph{Generalizing to combinatorial costs.}
With all the intuition developed so far, we then proceed to a significantly more general setting, where the principal's cost is a function of the combination of outcomes for each type, i.e., the principal's cost function is {\em combinatorial}.
% This further captures interactions between types, e.g., collaboration and competition,\todo{I'm not so convinced by this -- they don't directly compete or collaborate as in a multiagent mechanism?} as well as global constraints for the principal, e.g., budget or headcount constraints.
This further captures global constraints for the principal, e.g., budget or headcount constraints.
We present combinatorial counterparts of essentially all our results for additive costs in Section~\ref{sec:submodular}.

\subsection{Further Related Work}

Some recent research along the line of automated mechanism design includes designing auctions from observed samples \citep{cole2014sample,devanur2016sample,balcan2018general,gonczarowski2018sample}, mechanism design via deep learning \citep{duetting2019optimal,shen2019automated}, and estimating incentive compatibility \citep{balcan2019estimating}.
Most of these results focus on auctions, while in this paper, we consider automated mechanism design in a more general sense (though we focus mostly on the types of setting discussed in the introduction, which have more of a classification focus).
More closely related are results on automated mechanism design with partial verification \citep{auletta2011alternatives,yu2011mechanism,kephart2015complexity,kephart2016revelation}.
Those results are about conditions under which the revelation principle holds (including a relevant condition to our setting discussed later), and the computational complexity of deciding whether there exists an implementation of a {\em specific} mapping from types to outcomes.
On the other hand, we focus on algorithms for designing cost-{\em optimal} truthful mechanisms, which is largely orthogonal to those results.

Another closely related line of research is strategic machine learning.
There, a common assumption is that utility-maximizing agents can modify their features in some restricted way, normally at some cost \citep{hardt2016strategic,kleinberg2019classifiers,haghtalab2020maximizing,zhang2021incentive} (see also \citep{kephart2015complexity,kephart2016revelation}).
Strategic aspects of linear regression have also been studied \citep{perote2004strategy,dekel2010incentive,chen2018strategyproof}.
Our results differ from the above in that we study strategic classification from a more general point of view, and do not put restrictions on the class of classifiers or learning algorithms to be used.

Another line of work in economics considers mechanism design with costly misreporting, where the cost is unobservable to the principal \citep{laffont1986using,mcafee1987competition}.
These results are incomparable with ours, since they consider rather specific models, while we consider utility and cost functions of essentially any form.

\section{Additive Cost over Types}
\label{sec:additive}
Consider the classical setting of Bayesian (single-agent) mechanism design, which is as follows.
The agent can have one of many possible \emph{types}.
The agent reports a type to the principal (which may not be his true type), and then the principal chooses an \emph{outcome}.
The principal does not know the type of the agent, but she has a prior probability distribution over the agent's possible types.
The principal has a different cost for each combination of a type and an outcome.
The goal of the principal is to design a mechanism (a mapping from reports to outcomes) to minimize her expected cost assuming the agent best-responds to (i.e., maximizes his utility under) the mechanism.
% Equivalently, we can consider a setting with a population of agents, where the number of agents of type $i$ is proportional to type $i$'s prior probability in the aforementioned single-agent setting.
The principal aims to minimize her total cost over this population of agents, which is equal to the sum of her cost over individual agents.

In this section, we focus on the traditional setting where the principal's cost is additive over types.
In Section~\ref{sec:submodular}, we generalize our results to broader settings where the principal's cost function can be combinatorial (e.g., submodular) over types.

\paragraph{Notation.}
Let $\Theta$ be the agent's type space, and $\mathcal{O}$ the set of outcomes.
Let $n = |\Theta|$ and $m = |\mathcal{O}|$ be the numbers of types and outcomes respective.
Generally, we use $i \in \Theta$ to index types, and $j \in \mathcal{O}$ to index outcomes.
Let $\mathbb{R}_+ = [0, \infty)$.
We use $u_i: \mathcal{O} \to \mathbb{R}_+$ to denote the utility of a type $i$ agent, and $c_i: \mathcal{O} \to \mathbb{R}_+$ to denote the cost of the principal of assigning different outcomes to a type $i$ agent.

Let $R \subseteq \Theta \times \Theta$ denote all possible ways of misreporting, that is, a type $i$ agent can report type $i'$ if and only if $(i, i') \in R$.
We assume each type $i$ can always report truthfully, i.e., $(i, i) \in R$.
The principal specifies a (possibly randomized) mechanism $M: \Theta \to \mathcal{O}$, which maps reported types to (distributions over) outcomes.
The agent then responds to  maximize his expected utility under $M$.

Let $r_i$ denote the report of type $i$ when the agent best responds:
\[
    r_i \in \operatorname{argmax}_{i' \in \Theta, (i, i') \in R} \mathbb{E}[u_i(M(i'))].
\]
Without loss of generality, the principal's cost function can be scaled so that the prior distribution over possible types is effectively uniform.
The principal's cost under mechanism $M$ is then given by
\[
    % c(M) = \sum_{i \in \Theta} \mathbb{E}\left[c_i\left(M\left(\operatorname{argmax}_{i' \in \Theta:\, (i, i') \in R} \mathbb{E}[u_i(M(i'))]\right)\right)\right].
    c(M) = \sum_{i \in \Theta} \mathbb{E}\left[c_i\left(M\left(r_i\right)\right)\right]
\]
where both expectations are over the randomness in $M$.
Throughout the paper, given a set $S$, we use $\Delta(S)$ to denote the set of all distributions over $S$.

\subsection{Hardness without the Revelation Principle}
The well-known revelation principle states that when any type can report any other type, there always exists a truthful {\em direct-revelation} mechanism that is optimal for the principal.\footnote{A direct-revelation mechanism is a mechanism in which agents can only report their type, rather than sending arbitrary messages. A mechanism is truthful if it is always optimal for agents to report their true types.}
However, this is not true in the case of partial verification (see, e.g.,~\citep{green1986partially,yu2011mechanism,kephart2016revelation}).
In fact, it is known (see Theorem~4.10 of \citep{kephart2016revelation}) that in our setting, the revelation principle holds if and only if the reporting structure $R$ is transitive, i.e., for any types $i_1, i_2, i_3 \in \Theta$,
\[
    (i_1, i_2) \in R \text{ and } (i_2, i_3) \in R \implies (i_1, i_3) \in R.\footnote{
    To get some intuition for this characterization, suppose that $(i_1,i_2) \in R$,  $(i_2,i_3) \in R$, but  $(i_1,i_3) \notin R$, and we would like to accept $i_2$ and $i_3$ but not $i_1$.
    That is, higher types are better, and each type (except for the top one) can make itself look a bit, but not much, better than it is. 
    There is no truthful mechanism that achieves what we want: if we accept a report of $i_2$, we will end up accepting $i_1$ as well because it can misreport $i_2$. 
    On the other hand, if we accept only $i_3$, then we get what we want, by relying on $i_2$ to non-truthfully report $i_3$ (whereas $i_1$ cannot). 
    Hence, our goal can be achieved in a non-truthful implementation while it cannot be achieved in a truthful implementation, showing that the revelation principle does not hold in this case.}
\]
We begin our investigation by presenting a hardness result, which states that when the revelation principle does not hold, it is $\mathsf{NP}$-hard to find any optimal mechanism (even in the minimal nontrivial setting).

\begin{theorem}[$\mathsf{NP}$-hardness without the Revelation Principle]
\label{thm:np-hardness-without-revelation}
    When partial verification is allowed and the revelation principle does not hold, it is $\mathsf{NP}$-hard to find an optimal (randomized or deterministic) mechanism, even if there are only $2$ outcomes and all types share the same utility function.
\end{theorem}

We postpone the proof of Theorem~\ref{thm:np-hardness-without-revelation}, as well as all other proofs in this section, to Appendix~C.
In light of Theorem~\ref{thm:np-hardness-without-revelation}, in the rest of the paper, we focus on finding optimal truthful direct-revelation mechanisms.
That is, we consider only mechanisms $M$ where for any $(i_1, i_2) \in R$,
\[
    \mathbb{E}[u_{i_1}(M(i_1))] \ge \mathbb{E}[u_{i_1}(M(i_2))].
\]

\subsection{General vs. Structured Utility Functions}
Following the convention in the literature, we assume agents always break ties by reporting truthfully.
As a result, for a (possibly randomized) truthful mechanism $M$, the cost of the principal can be written as
\[
    c(M) = \sum_{i \in \Theta} \mathbb{E}[c_i(M(i))].
\]
Our first finding establishes a dichotomy between deterministic and randomized mechanisms when agents can have arbitrary utility functions.
On one hand, it is known that an optimal randomized mechanism can be found in polynomial time by formulating the problem as a linear program \citep{conitzer2002complexity}.
On the other hand, finding an optimal deterministic mechanism is $\mathsf{NP}$-hard even in an extremely simple setting as described below.

\begin{theorem}[$\mathsf{NP}$-hardness with General Utility Functions]
\label{thm:np-hardness-with-general-utility}
    When partial verification is allowed, even when the revelation principle holds, it is $\mathsf{NP}$-hard to find an optimal deterministic mechanism, even if there are only $3$ outcomes and the utility functions are single-peaked (see Appendix~B.1 for a definition).
    % When partial verification is allowed, restricted to truthful direct revelation mechanisms, it is $\mathsf{NP}$-hard to find an optimal deterministic mechanism, even if there are only $3$ outcomes and all types have single-peaked utility functions.
\end{theorem}

Although Theorem~\ref{thm:np-hardness-with-general-utility} establishes hardness for finding optimal deterministic mechanisms in most nontrivial cases, it leaves the possibility of efficient algorithms when all types have the same utility function --- which, as discussed in the introduction, is the setting we focus on in this paper.

\subsection{Finding Optimal Deterministic Mechanisms}
In light of the previously mentioned hardness results, for the rest of this section, we focus on the setting where the revelation principle holds and all types have the same utility function.

We recall and simplify some notations before we state the main result of this section (Theorem~\ref{thm:additive-deterministic}).
Let $u: \mathcal{O} \to \mathbb{R}_+$ be the common utility function of all types.
Recall that $n = |\Theta|$ is the number of types and $m = |\mathcal{O}|$ is the number of outcomes.
Let $\Theta = [n] = \{1, \ldots, n\}$.
For brevity, we use $\mathcal{O} = \{o_1, \dots, o_m\} \subseteq \mathbb{R}_+$ to encode the utility function $u$.
That is, for all $j \in [m]$, $o_j \in R_+$ is the utility of the agent under the $j$-th outcome.
Without loss of generality, assume $o_1 = 0$, and $o_j < o_{j + 1}$ for all $j \in [m - 1]$.

We give an efficient algorithm (Algorithm~\ref{alg:additive-deterministic}) for finding an optimal deterministic mechanism when partial verification is allowed.\footnote{
In a more empirically focused companion paper \citep{krishnaswamy2021classification}, we apply a simplified version of Algorithm~\ref{alg:additive-deterministic} to a special case of the problem studied in this paper.
There, the goal is to find a nearly optimal {\em binary} classifier (i.e., $m = 2$), given only {\em sample access} to the population distribution over the type space.
}
Our algorithm first builds a (capacitated) directed graph based on the principal's cost function and the reporting structure, then finds an $s$-$t$ min-cut in the graph, and then constructs a mechanism based on the found min-cut.
The idea is finite-capacity cuts in the graph constructed correspond bijectively to truthful mechanisms, where the capacity is precisely the cost of the principal.
In particular, we use edges with $\infty$ capacity to ensure that if one type gets an outcome, any type that can misreport the former must get at least as good an outcome.
See Figure~\ref{fig:flow} for an illustration of Algorithm~\ref{alg:additive-deterministic}.
The following theorem establishes the correctness and time complexity of Algorithm~\ref{alg:additive-deterministic}.

\begin{figure*}[ht]
\centering
\includegraphics[width=0.48\linewidth]{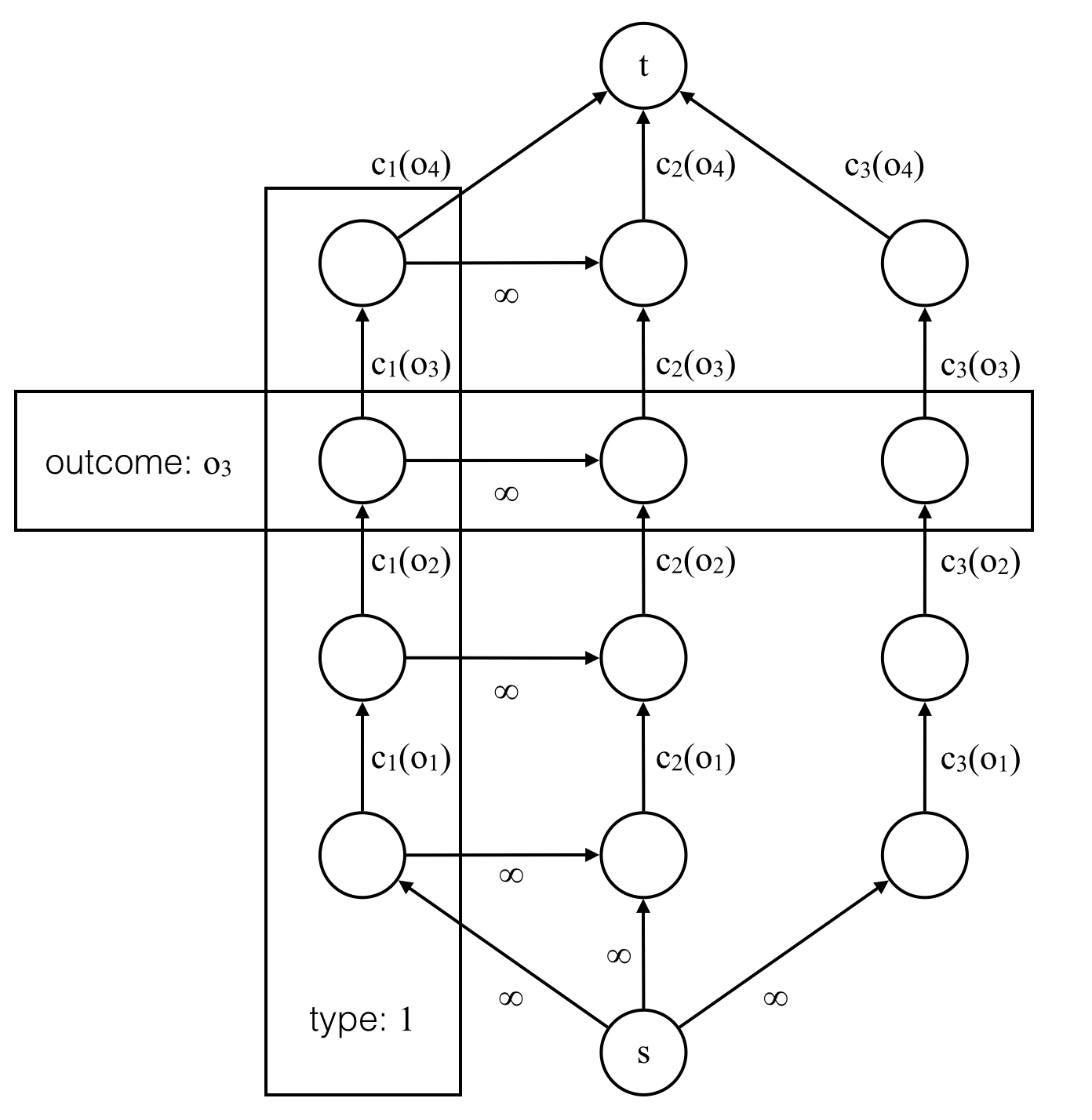}
\includegraphics[width=0.48\linewidth]{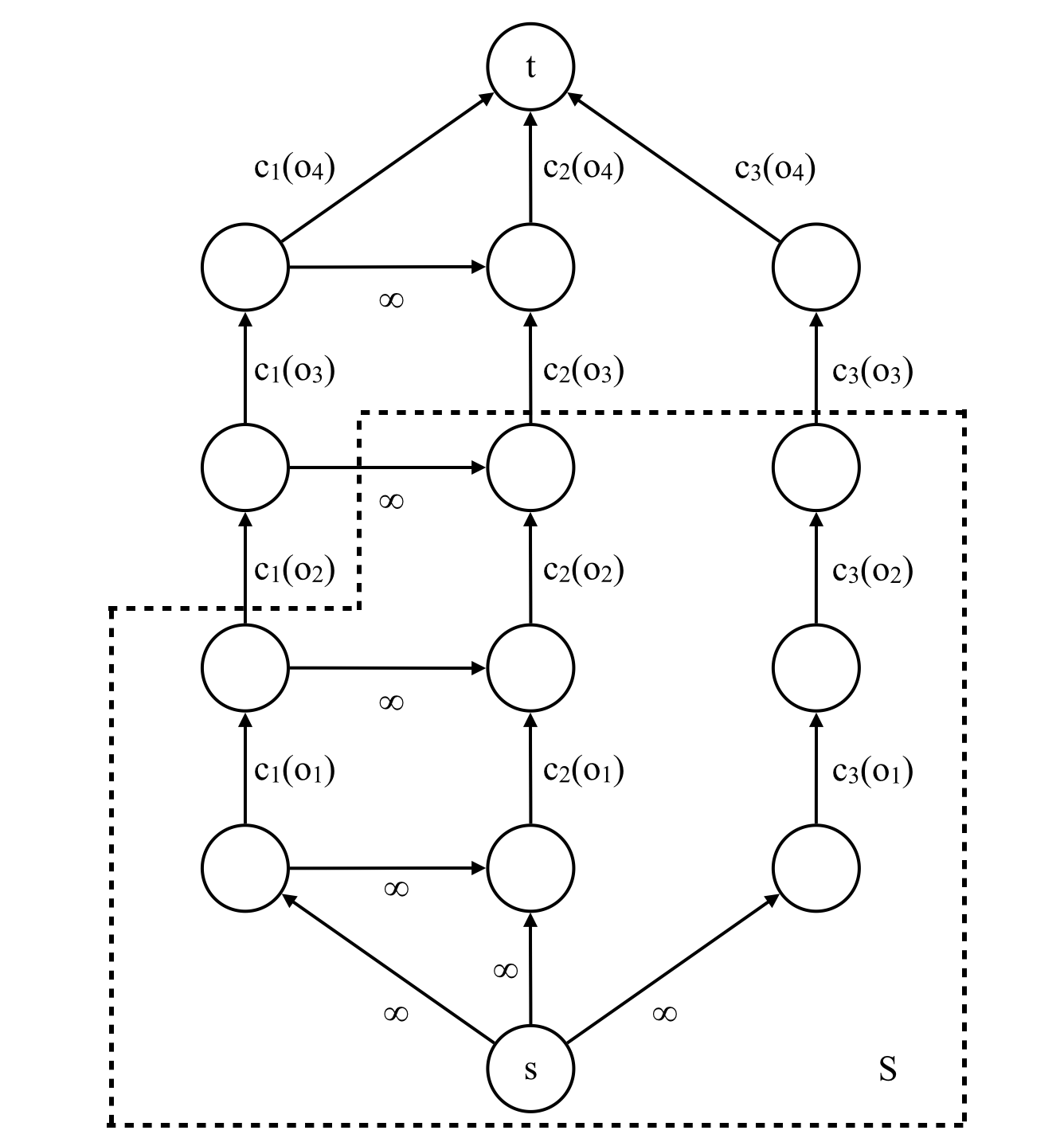}
\caption{
    An example of the graph constructed in Algorithm~\ref{alg:additive-deterministic}.
    As highlighted in the left graph, each row corresponds to an outcome and each column corresponds to a type.
    The horizontal edges with infinite capacity correspond to the fact that type $2$ can misreport as type $1$.
    The right graph gives a possible $s$-$t$ min-cut, which corresponds to a mechanism where $M(1) = o_2$, $M(2) = (o_3)$, and $M(3) = o_3$.
    The horizontal edges make sure that type $1$ never gets a more desirable outcome than type $2$, so type $2$ never misreports.
    The cost of the mechanism $M$ is equal to the value of the min-cut, which is $c_1(o_2) + c_2(o_3) + c_3(o_3)$.
}
\label{fig:flow}
\end{figure*}
% See Appendix~B.1 for an illustration of the graph constructed in Algorithm~\ref{alg:additive-deterministic}.

\begin{algorithm*}[t]
\SetAlgoNoLine
\KwIn{The set of types $\Theta$, the principal's cost function $\{c_i\}_{i \in \Theta}$ for each type, the set of outcomes $\mathcal{O}$ (which encodes the agents' common utility function), and the reporting structure $R$.}
\KwOut{A deterministic truthful mechanism $M: \Theta \to \mathcal{O}$ minimizing the principal's cost.}

Let $V \leftarrow (\Theta \times \mathcal{O}) \cup \{s, t\}$, $E \leftarrow \emptyset$\;
Replace $R$ with its transitive closure (using the Floyd–Warshall algorithm): \\
\For{$i_2, i_1, i_3 \in \Theta$ where $(i_1, i_2) \in R$ and $(i_2, i_3) \in R$}{
    $R \leftarrow R \cup \{(i_1, i_3)\}$ \;
}
\For{each type $i \in \Theta$}{
    $E \leftarrow E \cup \{(s, (i, o_1), \infty)\}$ (add an edge from $s$ to $(i, o_1)$ with capacity $\infty$) \;
    \For{each outcome $j \in [m - 1]$}{
        $E \leftarrow E \cup \{((i, o_j), (i, o_{j + 1}), c_i(o_j))\}$ (add an edge from $(i, o_j)$ to $(i, o_{j + 1})$ with capacity $c_i(o_j)$)\;
    }
    $E \leftarrow E \cup \{((i, o_m), t, c_i(o_m))\}$ (add an edge from $(i, o_m)$ to $t$ with capacity $c_i(o_m)$)\;
}
\For{each pair of types $(i_1, i_2)$ where $i_1 \ne i_2$ and $(i_1, i_2) \in R$, and each outcome $o_j \in \mathcal{O}$}{
    $E \leftarrow E \cup \{((i_2, o_j), (i_1, o_j), \infty)\}$ (add an edge from $(i_2, o_j)$ to $(i_1, o_j)$ with capacity $\infty$)\;
}
Compute an $s$-$t$ min-cut $(S, \overline{S})$ on graph $G = (V, E)$ \;
\For{each type $i \in \Theta$}{
    Let $M(i) = o_j$ where $j = \max\{j' \in [m] \mid (i, o_{j'}) \in S\}$\;
}
\Return $M$\;
\caption{Finding an optimal deterministic mechanism.}
\label{alg:additive-deterministic}
\end{algorithm*}

\begin{theorem}[Fast Algorithm for Finding Optimal Deterministic Mechanisms]
\label{thm:additive-deterministic}
    Suppose for any $i \in [n]$ and $j \in [m]$, $c_i(o_j) \in \mathbb{N}$.
    Let $W = \max_{i,j} c_i(o_j)$.
    Algorithm~\ref{alg:additive-deterministic} outputs an optimal deterministic truthful mechanism in time $O(T_\mathsf{MinCut}(mn, mn^2, W))$, where $T_\mathsf{MinCut}(n', m', W')$ is the time it takes to find an $s$-$t$ min-cut in a graph with $n'$ vertices, $m'$ edges, and maximum capacity $W'$.
\end{theorem}

We note that Algorithm~\ref{alg:additive-deterministic} only finds an optimal deterministic mechanism {\em subject to truthfulness} --- when the revelation principle does not hold, Algorithm~\ref{alg:additive-deterministic} may not find an unconditionally optimal mechanism (and indeed finding that is $\mathsf{NP}$-hard given Theorem~\ref{thm:np-hardness-without-revelation}).
The same applies for all our algorithmic results.

\begin{comment}
Before proceeding to our next result, we make two remarks regarding Algorithm~\ref{alg:additive-deterministic}.
\begin{itemize}
    \item For finding an optimal deterministic mechanism, the precise values of the agents' utility functions do not matter.
    Consequently, Algorithm~\ref{alg:additive-deterministic} works as long as all types {\em order} the outcomes in the same way.
    \item With minor modifications, Algorithm~\ref{alg:additive-deterministic} can handle {\em costly} misreporting, in which there is a fixed (non-negative) cost for type $i$ to report as type $i'$.
    Partial verification is a special case of costly misreporting: reporting either costs the agent $0$ or $\infty$, and the reporting structure $R$ is the set of all reporting actions which cost $0$.
    The key modification which allows Algorithm~\ref{alg:additive-deterministic} to handle costly misreporting is that the edges used to model the reporting structure can be diagonal (as opposed to horizontal), where the slope of the edge depends on each type's utility function and the cost of misreporting.
    We will not expand on this in the current paper.
\end{itemize}
\end{comment}

\subsection{Optimality of Deterministic Mechanisms with Convex Costs}
In the previous subsection, we showed that when the revelation principle holds and all types have the same utility function, there is a min-cut-based algorithm (Algorithm~\ref{alg:additive-deterministic}) that finds an optimal deterministic truthful mechanism.

In this subsection, we identify an important special case where there exists an optimal truthful mechanism that is deterministic (even when randomized mechanisms are allowed).
Consequently, we have an algorithm (Algorithm~\ref{alg:additive-deterministic}) for finding the optimal truthful mechanism that runs faster than solving a linear program.
More importantly, as we will show in Section~\ref{sec:additive-reduction}, we can essentially reduce the general case to this special case, and consequently obtain an algorithm for computing the optimal truthful mechanism whose runtime is asymptotically the same as Algorithm~\ref{alg:additive-deterministic}.

We first show (in Example~\ref{ex:gap}) that, in general, there can be an arbitrarily large gap between the cost of the optimal deterministic mechanism and that of the optimal randomized mechanism, even when restricted to truthful mechanisms and when all types share the same utility function.

\begin{example}[Gap between Deterministic and Randomized Mechanisms]
\label{ex:gap}
    There are $2$ types $\Theta = \{1, 2\}$ and $3$ outcomes $\mathcal{O} = \{o_1 = 1, o_2 = 2, o_3 = 3\}$, which encode the common utility function.
    The principal's cost is given by $c_1(o_1) = c_1(o_3) = \infty$, $c_1(o_2) = 0$, $c_2(o_1) = c_2(o_3) = 0$, and $c_2(o_2) = \infty$.
    The reporting structure $R$ allows any type to report any other type, i.e.,
    % \[
        $R = \{(1, 1), (2, 2), (1, 2), (2, 1)\}$.
    % \]
    Consider first the optimal truthful randomized mechanism, which as we argue below has cost $0$.
    To make the principal's cost finite, the optimal truthful mechanism must assign outcome $o_2$ to type $1$ with probability $1$, which gives type $1$ utility $2$.
    To prevent misreporting, the mechanism must give type $2$ the same expected utility.
    And again, to make the cost finite, it must never assign outcome $o_2$ to type $2$.
    The unique way to satisfy the above is to assign to type $2$  outcome $o_1$ with probability $1/2$, and $o_3$ with probability $1/2$.
    
    Now consider any deterministic truthful mechanism.
    Any truthful mechanism must give both types the same utility to prevent misreporting.
    The only way to achieve this deterministically is to assign the same outcome to both types.
    However, all $3$ possibilities result in infinite total cost, so all deterministic truthful mechanisms have cost infinity.
\end{example}

Example~\ref{ex:gap} shows that Algorithm~\ref{alg:additive-deterministic} in general does not find an (approximately) optimal truthful mechanism when randomized mechanisms are allowed.
In such cases, one has to fall back to significantly slower algorithms, e.g., solving the straightforward LP formulation of the problem with $mn$ variables and $n^2$ constraints.
It is worth noting that the LP formulation does not utilize the fact that types share an identical utility function.
To address this issue, we identify an important special case where there does exist an optimal truthful mechanism that is deterministic: when the principal's cost is convex in the common utility function.
More importantly, as we will show in Section~\ref{sec:additive-reduction}, we can reduce the problem of finding the optimal randomized mechanism under general costs to the problem of finding the optimal mechanism with convex costs.
First we formally define the notion of convex costs we use.

\begin{definition}[Convex Costs]
\label{def:convex-costs}
    For any $i \in \Theta$, let the piecewise linear extension $c_i^\ell: [o_1, o_m] \to \mathbb{R}_+$ of $c_i$ be such that (1) for any $x \in \mathcal{O}$, $c_i^\ell(x) = c_i(x)$, and (2) for any $x \in [o_1, o_m] \setminus \mathcal{O}$,
    \[
        c_i^\ell(x) = \frac{o_{j + 1} - x}{o_{j + 1} - o_j} \cdot c_i(o_j) + \frac{x - o_j}{o_{j + 1} - o_j} c_i(o_{j + 1}),
    \]
    where $j = \max\{j' \in [m] \mid o_{j'} \le x\}$.
    The principal's cost function $\{c_i\}_{i \in \Theta}$ is convex if for every $i \in \Theta$, the piecewise linear extension $c_i^\ell$ of $c_i$ is convex.
\end{definition}

\begin{lemma}[Optimality of Deterministic Mechanisms with Convex Costs]
\label{lem:additive-convex}
    When all types share the same utility function, and the principal's cost function is convex, there is an optimal truthful mechanism that is deterministic even with partial verification allowed.
    %\todo{should we repeat that the revelation principle holds?}
\end{lemma}

\begin{comment}
We make a few remarks regarding the above lemma (Lemma~\ref{lem:additive-convex}).
\begin{itemize}
    \item The proof we present above is a combination of several concrete arguments.
    There is an alternative relatively high-level, and sometimes more useful, interpretation of the lemma, which is based on a convex program formulation of the problem.
    We will make heavy use of this alternative interpretation in the rest of the paper, especially when dealing with randomized mechanisms.
    \item Throughout the paper we assume payments are not allowed.
    One may show that with payments, there always exists an optimal truthful mechanism that is deterministic, as long as both agents and the principal value payments linearly.
    Moreover, there exist relatively simple algorithms for computing an optimal mechanism with payments.
    We will not expand on this in the current paper.
\end{itemize}
\end{comment}

\subsection{Reducing General Costs to Convex Costs}
\label{sec:additive-reduction}
Lemma~\ref{lem:additive-convex} together with Algorithm~\ref{alg:additive-deterministic} provides an efficient way for finding optimal truthful mechanisms with convex costs (even when randomized mechanisms are allowed).
One may still wonder if it is possible to design faster algorithms {\em in general} than solving the standard LP formulation, presumably by exploiting the additional structure that the agents share the same utility function.
To this end, we observe that for computing optimal mechanisms, only the convex envelope of the principal's cost function matters.
Given this observation, we show that finding optimal truthful mechanisms can be reduced very efficiently to finding optimal deterministic mechanisms.
%\todo{in alg 2 should we maybe say $\delta \in \Delta(O)$ instead and then in the subscripts of the expectation operators that $o$ is drawn according to $\delta$?}

\begin{algorithm}[t]
\SetAlgoNoLine
\KwIn{The set of types $\Theta$, the principal's cost function $\{c_i\}_{i \in \Theta}$ for each type, the set of outcomes $\mathcal{O}$ (which encodes the common utility function), and the reporting structure $R$.}
\KwOut{A truthful mechanism $M: \Theta \to \mathcal{O}$ minimizing the principal's cost.}

\For{each type $i$}{
    Compute the convex envelope $c_i^-: [o_1, o_m] \to \mathbb{R}_+$ of $c_i$, defined such that for any $x \in [o_1, o_m]$,
    \[
        c_i^-(x) = \min_{o \in \Delta(\mathcal{O}), \mathbb{E}[o] = x}\mathbb{E}[c_i(o)].
    \]
    Let $\widehat{c_i}$ be $c_i^-$ restricted to $\mathcal{O}$\;
}
Run Algorithm~\ref{alg:additive-deterministic} on input $(\Theta, \{\widehat{c_i}\}_{i \in \Theta}, \mathcal{O}, R)$.
Let $\widehat{M}$ be the resulting deterministic mechanism\;
\For{each type $i$}{
    $\displaystyle M(i) \leftarrow \operatorname{argmin}_{o \in \Delta(\mathcal{O}), \mathbb{E}[o] = \widehat{M}(i)} \mathbb{E}[c_i(o)]$ \;
}
\Return $M$\;
\caption{Finding an optimal (possibly randomized) truthful mechanism.}
\label{alg:additive-randomized}
\end{algorithm}

We present Algorithm~\ref{alg:additive-randomized}, which computes the optimal truthful mechanism and has the same asymptotic runtime as Algorithm~\ref{alg:additive-deterministic}.
Algorithm~\ref{alg:additive-randomized} first computes the convex envelope of the principal's cost function, and then finds an optimal ``deterministic'' mechanism by calling Algorithm~\ref{alg:additive-deterministic} with the same types and outcomes, but replacing the principal's cost function with its convex envelope.
Algorithm~\ref{alg:additive-randomized} then recovers an optimal randomized mechanism from the ``deterministic'' one, by interpreting each ``deterministic'' outcome as a convex combination of outcomes in an optimal way.
The following theorem establishes the correctness and time complexity of Algorithm~\ref{alg:additive-randomized}.

\begin{theorem}
\label{thm:additive-randomize} 
    Algorithm~\ref{alg:additive-randomized} finds an optimal (possibly randomized) truthful mechanism, in asymptotically the same time as Algorithm~\ref{alg:additive-deterministic}.
\end{theorem}

Below we give a comparison between the time complexity of our algorithm, Algorithm~\ref{alg:additive-randomized}, and that of the LP-based approach.\footnote{We note that a conclusive comparison is unrealistic since algorithms for both LP and min-cut keep being improved.}
The current best algorithm for LP \citep{cohen2019solving} takes time that translates to $\widetilde{O}(n^{2.37}m^{2.37} + n^{4.74})$\footnote{$\widetilde{O}$ hides a poly-logarithmic factor.} in our setting (this is, for example, at least $\widetilde{O}(n^{3.24} m^{1.5})$).
The current best algorithm for $s$-$t$ min-cut \citep{lee2014path} takes time that translates to $\widetilde{O}(n^{2.5} m^{1.5})$ in our setting.
Moreover, in a typical classification setting, it is the number of outcomes (corresponding to ``accept'', etc.) $m$ that is small, and the number of types (e.g., ``(CS major, highly competitive, female, international, \dots)'', ``(math major, acceptable, male, domestic, \dots)'') $n$ is much larger.
In such cases, the improvement becomes even more significant.
Our results are theoretical, but practically, while there are highly optimized packages for LP, there are also highly optimized packages for max-flow / min-cut that are still much faster.
Last but not least, in many practical settings, the principal has to implement a deterministic policy (it is hard to imagine college admissions explicitly made random), in which case our Algorithm~\ref{alg:additive-deterministic} can be applied while LP generally does not give a solution.

\begin{comment}
Finally, we make a few remarks regarding Algorithm~\ref{alg:additive-randomized}.
\begin{itemize}
    \item Algorithm~\ref{alg:additive-randomized} gives a constructive proof that finding an optimal truthful mechanism is always no harder than finding an optimal truthful deterministic mechanism with convex costs.
    As a result, a faster algorithm for the latter problem would imply a faster algorithm for the former.
    \item As a byproduct, Algorithm~\ref{alg:additive-randomized} shows that in general, to achieve the minimum cost, it suffices to randomize only between two outcomes for each type, .
\end{itemize}
\end{comment}

\section{Generalizing to Combinatorial Costs}
\label{sec:submodular}

In this section, we generalize the problem considered in the previous section, allowing the principal to have a combinatorial cost function over outcomes for each type.
% The problem studied in the previous section can be viewed as a special case of this general problem.
See Appendix~A for a more detailed exposition.

\paragraph{The combinatorial setting.}
As before, let $\Theta = [n]$ be the set of types, $\mathcal{O} = \{o_j\}_{j \in [m]} \subseteq \mathbb{R}_+$ be the set of outcomes encoding the common utility function, and $R \subseteq \Theta \times \Theta$ be the reporting structure.
The principal's cost function $c: \mathcal{O}^\Theta \to \mathbb{R}_+$ now maps a vector $O = (O^i)_i$ of outcomes for all types to the principal's cost $c(O)$.
This subsumes the additive case, since one can set the cost function $c$ to be
\[
    c((O^i)_i) = \sum_{i \in \Theta} c_i(O^i).
\]
Because the cost function is now combinatorial, it matters how the mechanism combines outcomes for different types.
We therefore modify the definition of a randomized mechanism $M \in \Delta(\Theta \to \mathcal{O}) = \Delta(\mathcal{O}^\Theta)$, so that it allows correlation across different types.
The principal's cost from using a truthful mechanism $M$ is then
% \[
    $c(M) = \mathbb{E}[c((M(i))_i)]$.
% \]
% We treat $M$ as a distribution or a random variable over $\mathcal{O}^\Theta$ interchangeably.
% \todo{not sure that preceding sentence adds much}
% Note that each type's utility is still independent of what other types get.
For type $i$, the utility from executing mechanism $M$ is still
% \[
    $u_i(M) = \mathbb{E}[M(i)]$.
% \]
$M$ is truthful iff for any $(i_1, i_2) \in R$,
% \[
    $u_{i_1}(M) \ge u_{i_2}(M)$.
% \]
In the rest of the section, we present combinatorial generalizations of all our algorithmic and structural results given in the previous section.

\paragraph{General vs.\ submodular cost functions.}
Combinatorial functions in general are notoriously hard to optimize, even ignoring incentive issues.
To see the difficulty, observe that a combinatorial cost function $c: \mathcal{O}^\Theta \to \mathbb{R}_+$ over $\mathcal{O}^\Theta$ generally does not even admit a succinct representation (e.g., one whose size is polynomial in $m$ and $n$).
It is therefore infeasible to take the entire cost function as input to an algorithm.
To address this issue, the standard assumption in combinatorial optimization is that algorithms can access the combinatorial function through {\em value queries}.
That is, we are given an oracle that can evaluate the combinatorial function $c$ at any point $O \in \mathcal{O}^\Theta$, obtaining the value $c(O)$ in constant time.
For the rest of the paper, we assume that our algorithm can  access the cost function only through value queries.

Still, in order to minimize an {\em arbitrary} combinatorial function, in general one needs $\Omega(m^n)$ queries to obtain any nontrivial approximation.
Despite that, there exist efficient algorithms for combinatorial minimization for an important subclass of cost functions, namely submodular functions.

\begin{definition}[Submodular Functions]
    For any $O_1 = (O_1^i)_i \in \mathcal{O}^\Theta$ and $O_2 = (O_2^i)_i \in \mathcal{O}^\Theta$, let
    \[
        O_1 \wedge O_2 = (\min(O_1^i, O_2^i))_i \ \text{and} \ O_1 \vee O_2 = (\max(O_1^i, O_2^i))_i.
    \]
    A combinatorial cost function $c: \mathcal{O}^\Theta \to \mathbb{R}_+$ is submodular if for any $O_1, O_2 \in \mathcal{O}^\Theta$,
    \[
        c(O_1) + c(O_2) \ge c(O_1 \wedge O_2) + c(O_1 \vee O_2).
    \]
\end{definition}

In the rest of this section, we focus on submodular cost functions.
For this important special case, we give efficient algorithms for finding optimal truthful deterministic / randomized mechanisms, as well as a sufficient condition for the existence of an optimal mechanism that is deterministic.

\paragraph{Finding optimal deterministic mechanisms.}
First we present a polynomial-time combinatorial algorithm for finding optimal truthful deterministic mechanisms with partial verification, when the cost function is submodular.

\begin{theorem}
% \label{thm:submodular-deterministic}
    There exists a polynomial-time algorithm which accesses the cost function via value queries only, and computes an optimal deterministic truthful mechanism when partial verification is allowed and the cost function is submodular.
\end{theorem}

\paragraph{Sufficient condition for the optimality of deterministic mechanisms.}
% We have shown in Example~\ref{ex:gap} that the gap between deterministic and randomized mechanisms can be arbitrarily large.
Restricted to additive cost functions, Lemma~\ref{lem:additive-convex} gives a sufficient condition under which there exists an optimal mechanism that is deterministic.
We present below a combinatorial version of this structural result when the outcome space is binary, i.e., when $m = 2$.

\begin{theorem}[Optimality of Deterministic Mechanisms with Binary Outcomes]
% \label{lem:submodular-binary}
    When the outcome space is binary, i.e., $|\mathcal{O}| = 2$, and the principal's cost function is submodular, there is an optimal truthful mechanism that is deterministic, even when partial verification is allowed.
\end{theorem}

\paragraph{Computing optimal randomized mechanisms.}
Finally we give an algorithm for finding an optimal mechanism with arbitrary submodular cost functions.

\begin{theorem}
% \label{thm:submodular-randomized}
    When the cost function $c$ is submodular and bounded, for any desired additive error $\varepsilon > 0$, there is an algorithm which finds an $\varepsilon$-approximately optimal (possibly randomized) truthful mechanism\footnote{An $\varepsilon$-approximately optimal truthful mechanism is a truthful mechanism whose expected cost is at most $\varepsilon$ larger than the minimum possible cost of any truthful mechanism.} in time $\mathrm{poly}(n, m, \log(1 / \varepsilon))$, even if partial verification is allowed.
\end{theorem}

\begin{comment}
\section{Conclusion and Future Research}

In this paper, we systematically study the problem of automated mechanism design with partial verification, when all types share the same preference.
Our results provide a relatively complete picture for automated mechanism design with partial verification, and imply, among other applications, practically meaningful approaches to designing optimal strategy-proof classification mechanisms.
Moreover, we generalize the problem to the combinatorial domain, and show that polynomial-time algorithms exist whenever the principal's cost function is submodular.
Conceptually, our results establish a curious connection between automated mechanism design and convex optimization, which may be of broader interest.
Future research directions include more general conditions for the existence of optimal deterministic mechanisms, especially with combinatorial costs.
Another interesting direction is automated mechanism design with costly misreporting, which generalizes the partial verification setting even further.
\end{comment}

\clearpage

\section*{Ethics Statement}
Our results can be used to encourage truthful reporting, and thereby improve the efficiency of mechanisms for, e.g., allocating public resources.
In particular, this helps the principal to allocate resources to agents in actual need.
By presenting more efficient algorithms, we make large-scale applications of automated mechanism design possible, which can be applied to problems related to social good that were previously beyond reach.
Of course, our algorithms, as well as any other algorithm, could potentially cause harm if implemented in an irresponsible way (e.g., by using a cost function that discriminates between applicants in an unfair way).

\section*{Acknowledgements}

Part of this work was done while Yu Cheng was visiting the Institute of Advanced Study.
Hanrui Zhang and Vincent Conitzer are thankful for support from NSF under award IIS-1814056.

{
\bibliographystyle{plainnat}
\bibliography{ref}
}

\clearpage

\appendix

\section{Generalizing to Combinatorial Costs}
\label{app:submodular}

In this section, we generalize the problem considered in the previous section, allowing the principal to have a combinatorial cost function over outcomes for each type.
The problem studied in the previous section can be viewed as a special case (where the principal's cost function is additive over types) of this general problem.
Before we proceed to the formal definition of the problem, to better motivate combinatorial cost functions, consider the following example.

\begin{example}
\label{ex:comb-c}
    Suppose in addition to an additive cost function $\{c_i\}_{i \in \Theta}$, the principal has to pay an overhead cost $c_0 > 0$ as long as any type receives a nontrivial outcome, i.e., if there exists $i \in \Theta$, such that $M(i) \in \mathcal{O} \setminus \{o_1\}$.
    In such cases, the principal's overall cost from executing a truthful deterministic mechanism $M$ can be written as
    \[
        c(M) = \sum_{i \in \Theta} c_i(M(i)) + c_0 \cdot \mathbb{I}[\exists i \in \Theta: M(i) \ne o_1],
    \]
    where $\mathbb{I}[\cdot]$ is the indicator of a statement.
\end{example}

In the above rather natural example, the principal's cost is no longer additive over types.
As a result, there is no way to properly formulate Example~\ref{ex:comb-c} using our previous definitions.
We generalize the principal's cost function, as well as the definition of mechanisms, as follows.

\paragraph{Notation.}
As before, let $\Theta = [n]$ be the set of types, $\mathcal{O} = \{o_j\}_{j \in [m]} \subseteq \mathbb{R}_+$ be the set of outcomes encoding the common utility function, and $R \subseteq \Theta \times \Theta$ be the reporting structure.
The principal's cost function $c: \mathcal{O}^\Theta \to \mathbb{R}_+$ now maps a vector $O = (O^i)_i$ of outcomes for all types to the principal's cost $c(O)$.
This subsumes the additive case, since one can set the cost function $c$ to be
\[
    c((O^i)_i) = \sum_{i \in \Theta} c_i(O^i).
\]
Because the cost function is now combinatorial, it matters how the mechanism combines outcomes for different types.
We therefore modify the definition of a (possibly randomized) mechanism $M \in \Delta(\Theta \to \mathcal{O}) = \Delta(\mathcal{O}^\Theta)$, such that it allows correlation across different types.
The principal's cost from executing a truthful mechanism $M$ is then
\[
    c(M) = \mathbb{E}[c((M(i))_i)].
\]
We treat $M$ as a distribution or a random variable over $\mathcal{O}^\Theta$ interchangeably.
Note that each type's utility is still independent of what other types get.
So for type $i$, the utility from executing mechanism $M$ is still
\[
    u_i(M) = \mathbb{E}[M(i)].
\]
And $M$ is truthful iff for any $(i_1, i_2) \in R$,
\[
    u_{i_1}(M) \ge u_{i_2}(M).
\]

In the rest of the section, we present combinatorial generalizations of all our algorithmic and structural results given in the previous section.

\subsection{General vs.\ Submodular Cost Functions}
Combinatorial functions in general are notoriously hard to optimize, even ignoring incentive issues.
To see the difficulty, observe that a combinatorial cost function $c: \mathcal{O}^\Theta \to \mathbb{R}_+$ over $\mathcal{O}^\Theta$ generally does not even admit a succinct representation (e.g., one whose size is polynomial in $m$ and $n$).
It is therefore infeasible to take the entire cost function as input to an algorithm.

To address this issue, the standard assumption in combinatorial optimization is that algorithms can access the combinatorial function through {\em value queries}.
That is, we are given an oracle that can evaluate the combinatorial function $c$ at any point $O \in \mathcal{O}^\Theta$, obtaining the value $c(O)$ in constant time.
For the rest of the paper, we assume that our algorithm can only access the cost function only through value queries.

Now suppose we are to design an algorithm to minimize an arbitrary combinatorial cost function, without any additional constraint.
That is, given a combinatorial cost function $c: \mathcal{O}^\Theta \to \mathbb{R}_+$, we wish to find a point $O \in \mathcal{O}^\Theta$, such that $c(O)$ is minimized over $\mathcal{O}^\Theta$.
The example below shows that any algorithm which interacts with $c$ only through value queries needs $\Omega(m^n)$ queries to obtain any nontrivial approximation to the above seemingly basic problem.

\begin{example}
    Let the cost function $c$ be generated in the following random way.
    A point $O^*$ is drawn from $\mathcal{O}^\Theta$ uniformly at random.
    $c$ is then constructed such that $c(O^*) = 0$, and $c(O) = 1$ for any $O \ne O^*$.
    To minimize $c$, the algorithm has to find $O^*$.
    This is equivalent to guessing a uniformly random number among $m^n$ numbers.
    To guess $O^*$ successfully with constant probability, one has to make $\Omega(m^n)$ guesses.
\end{example}

The above issue has been identified in combinatorial optimization since decades ago.
Despite the fact that general combinatorial cost functions are hard to minimize, researchers have developed efficient algorithms for combinatorial minimization for an important subclass of cost functions, namely submodular functions, defined below.

\begin{definition}[Submodular Functions]
    For any $O_1 = (O_1^i)_i \in \mathcal{O}^\Theta$ and $O_2 = (O_2^i)_i \in \mathcal{O}^\Theta$, let
    \[
        O_1 \wedge O_2 = (\min(O_1^i, O_2^i))_i \quad \text{and} \quad O_1 \vee O_2 = (\max(O_1^i, O_2^i))_i.
    \]
    A combinatorial cost function $c: \mathcal{O}^\Theta \to \mathbb{R}_+$ is submodular, if for any $O_1, O_2 \in \mathcal{O}^\Theta$,
    \[
        c(O_1) + c(O_2) \ge c(O_1 \wedge O_2) + c(O_1 \vee O_2).
    \]
\end{definition}

In the rest of the section, we focus on submodular cost functions.
For this important special case, we give efficient algorithms for finding optimal truthful deterministic/randomized mechanisms, as well as a sufficient condition for the existence of an optimal mechanism that is deterministic.

\subsection{Finding Optimal Deterministic Mechanisms}

First we present a polynomial-time combinatorial algorithm for finding optimal truthful deterministic mechanisms with partial verification, when the cost function is submodular.
The algorithm is based on the key observation that the space of truthful deterministic mechanisms is a distributive lattice (defined below in Lemma~\ref{lem:distributive-lattice}).
Given this observation, it is known that the problem can be reduced to minimizing a submodular function without additional constraints, which can be solved efficiently.

\begin{lemma}
\label{lem:distributive-lattice}
    Fix the set of types $\Theta$, the set of outcomes $\mathcal{O}$, and the reporting structure $R$.
    Let $\mathcal{T} \subseteq \mathcal{O}^\Theta$ be the space of all possible ways of assigning outcomes to types, such that no type has the incentive to misreport.
    That is,
    \[
        \mathcal{T} = \{O = (O^i)_i \in \mathcal{O}^\Theta \mid \forall (i_1, i_2) \in R,\, O^{i_1} \ge O^{i_2}\}.
    \]
    Then $\mathcal{T}$ is a distributive lattice, i.e., $\mathcal{T}$ satisfies the following conditions.
    \begin{itemize}
        \item For any $O_1, O_2 \in \mathcal{T}$, $O_1 \wedge O_2 \in \mathcal{T}$, and $O_1 \vee O_2 \in \mathcal{T}$.
        \item For any $O_1, O_2, O_3 \in \mathcal{T}$,
            $O_1 \vee (O_2 \wedge O_3) = (O_1 \vee O_2) \wedge (O_1 \vee O_3)$.
    \end{itemize}
\end{lemma}

\begin{proof}%[Proof of Lemma~\ref{lem:distributive-lattice}]
    Consider the first property.
    Fix any $O_1, O_2 \in \mathcal{T}$, and let $O_- = O_1 \wedge O_2$, and $O_+ = O_1 \vee O_2$.
    For any $(i_1, i_2) \in R$, since $O_1^{i_1} \ge O_1^{i_2}$ and $O_2^{i_1} \ge O_2^{i_2}$, we have
    \begin{align*}
        O_-^{i_1} & = \min(O_1^{i_1}, O_2^{i_1}) \ge \min(O_1^{i_2}, O_2^{i_2}) = O_-^{i_2}.
    \end{align*}
    This implies $O_- \in \mathcal{T}$.
    Similarly we may show $O_+ \in \mathcal{T}$.
    In other words, the first property holds.

    For the second property, simply consider the $i$-th coordinate for any $i \in \Theta$.
    Fix $O_1, O_2, O_3 \in \mathcal{T}$, we have
    \[
        \max(O_1^i, \min(O_2^i, O_3^i)) = \min(\max(O_1^i, O_2^i), \max(O_1^i, O_3^i)).
    \]
    Since this is true for any $i$, the second property follows immediately.
\end{proof}

% We postpone the proof of Lemma~\ref{lem:distributive-lattice}, as well as all other proofs in this section, to Appendix~\ref{app:submodular}.
Given Lemma~\ref{lem:distributive-lattice}, we can apply the algorithm and the reduction by \citet{schrijver2000combinatorial} to obtain an efficient algorithm directly.

\begin{corollary}
\label{cor:submodular-deterministic}
    There exists a polynomial-time algorithm which accesses the cost function via value queries only, and computes an optimal deterministic truthful mechanism when partial verification is allowed and the cost function is submodular.
\end{corollary}

\begin{proof}%[Proof of Corollary~\ref{cor:submodular-deterministic}]
    Let $\mathcal{T}$ be the family of truthful assignments defined in Lemma~\ref{lem:distributive-lattice}.
    The problem of finding an optimal deterministic truthful mechanism can be equivalently formulated as the following optimization problem.
    \[
        \min_{O \in \mathcal{T}} c(O).
    \]
    This can be solved by applying the reduction in Section~6 of \cite{schrijver2000combinatorial}\footnote{Although the reduction presented therein is for ring families, one may check it also works for distributive lattices.} to any algorithm for minimizing submodular functions (e.g., the one given in \cite{schrijver2000combinatorial}).
\end{proof}

We remark that the algorithm by \citet{schrijver2000combinatorial} can be applied as well in the classical additive setting, but due to its generality, is significantly less efficient than Algorithm~\ref{alg:additive-deterministic}.

\subsection{Sufficient Condition for the Optimality of Deterministic Mechanisms}
We have shown in Example~\ref{ex:gap} that the gap between deterministic and randomized mechanisms can be arbitrarily large.
Restricted to additive cost functions, Lemma~\ref{lem:additive-convex} gives a sufficient condition under which there exists an optimal mechanism that is deterministic.
We present in this subsection a combinatorial version of this structural result when the outcome space is binary, i.e., when $m = 2$.

\begin{lemma}[Optimality of Deterministic Mechanisms with Binary Outcomes]
\label{lem:submodular-binary}
    When the outcome space is binary, i.e., $|\mathcal{O}| = 2$, and the principal's cost function is submodular, there is an optimal truthful mechanism that is deterministic, even when partial verification is allowed.
\end{lemma}

\begin{proof}%[Proof of Lemma~\ref{lem:submodular-binary}]
    The overall plan is similar to that of the proof of Lemma~\ref{lem:additive-convex}.
    We begin with a (possibly randomized) optimal truthful mechanism $M$, and show that without loss of generality, we may assume its support has some monotone structure.
    We then round this mechanism, such that the resulting deterministic mechanism is always truthful, and the expected cost of the rounded mechanism is equal to the cost of $M$.

    Without loss of generality, suppose $\mathcal{O} = \{0, 1\}$.
    Observe that $\mathcal{O}^\Theta$ is isomorphic to $2^\Theta$, so in the rest of the proof, we interchangeably represent an outcome vector $O$ as a subset of $\Theta$, i.e., the set
    \[
        \{i \in \Theta \mid O^i = 1\}.
    \]
    Let $M$ be any (possibly randomized) optimal truthful mechanism.
    Below we treat $M$ as a random variable distributed over $\mathcal{O}^\Theta$, or interchangeably, a random subset of $\Theta$.
    For any $O \in \mathcal{O}^\Theta$, let $p(O) = \Pr[M = O]$ be the probability that $M$ assigns outcomes $O$ to types.
    We further require $M$ to maximize the following potential function among all optimal truthful mechanisms.
    \[
        \mathbb{E}[|M|^2] = \sum_{O \in \mathcal{O}^\Theta} p(O) \cdot |O|^2.
    \]
    We argue below that for such an $M$, no two outcome vectors in the support of $M$ ``cross.''

    We say two outcome vectors $O_1$ and $O_2$ (represented as sets) cross, if $O_1 \not\subseteq O_2$ and $O_2 \not\subseteq O_1$.
    Toward a contradiction, suppose $O_1$ and $O_2$ cross, where without loss of generality $p(O_1) \ge p(O_2) > 0$.
    Let $p = p(O_2)$.
    We show that moving probability mass from $O_1$ and $O_2$ to $O_1 \cap O_2$ and $O_1 \cup O_2$ simultaneously preserves truthfulness, does not increase the cost, and strictly increases the potential of $M$, which contradicts the choice of $M$.

    To be precise, we decrease $p(O_1)$ and $p(O_2)$ simultaneously by $p$, and increase $p(O_1 \cap O_2)$ and $p(O_1 \cup O_2)$ simultaneously by $p$.
    To see why truthfulness is preserved, observe that the expected utility of any type does not change after the modification.
    The change of the cost can be written as
    \begin{align*}
        \Delta c(M) & = p \cdot (c(O_1 \cap O_2) + c(O_1 \cup O_2) - c(O_1) - c(O_2)) \\
        & \le p \cdot (c(O_1) + c(O_2) - c(O_1) - c(O_2)) \tag{submodularity of $c$} \\
        & = 0,
    \end{align*}
    so the cost does not increase.
    Finally, the change of the potential is
    \[
        p \cdot (|O_1 \cap O_2|^2 + |O_1 \cup O_2|^2 - |O_1|^2 - |O_2|^2).
    \]
    It is easy to check the above is strictly positive as long as $O_1$ and $O_2$ cross.

    From now on we assume no two outcomes in the support of $M$ cross, or equivalently, the support of $M$ is a family of nested subsets of $\Theta$.
    For any $r \in [0, 1]$, let $M_r \subseteq \Theta$ be such that
    \[
        M_r = \{i \in \Theta \mid u_i \ge r\}.
    \]
    Observe that $M_r$ is truthful for any $r \in [0, 1]$.
    In fact, for any $(i_1, i_2) \in R$, we always have
    \[
        i_2 \in M_r \Longrightarrow i_1 \in R.
    \]
    Truthfulness then follows.

    Consider the random (but not randomized) mechanism $M_r$ when $r$ is uniformly distributed over $[0, 1]$.
    We argue below that the expected cost of $M_r$ is precisely that of $M$.
    In fact, $M_r$ and $M$ are even identically distributed.
    For any $i \in \Theta$, let the expected utility of type $i$ be $u_i = \Pr[i \in M]$.
    Without loss of generality, suppose $M$ satisfies $u_i \ge u_{i + 1}$ for any $i \in [n - 1]$.
    To show $M_r$ and $M$ are identically distributed, we only need to show that
    \[
        p(O) > 0 \Longrightarrow \exists i \in \Theta, O = [i].
    \]
    In other words, $M$ only assigns outcomes that are prefixes of $\Theta$.
    Given this, $M_r$ is the only distribution that gives each type $i$ expected utility $u_i$ simultaneously.

    To see why the above claim is true, suppose there exists some $O \subseteq \Theta$ where $1 \notin O$ and $p(O) > 0$.
    Then since the support of $M$ is a nested family of subsets of $\Theta$, for any $i \in O$, regardless of the realization of $M$, we always have
    \[
        1 \in M \Longrightarrow i \in O.
    \]
    As a result,
    \[
        u_i = \sum_{O' \subseteq \Theta: i \in O'} p(O') \ge p(O) + \sum_{O' \subseteq \Theta: 1 \in O'} p(O') = p(O) + u_1 > u_1,
    \]
    a contradiction.
    This concludes the proof.
\end{proof}

We make the following remarks regarding Lemma~\ref{lem:submodular-binary}.
\begin{itemize}
    \item The binary outcomes assumption, despite being more restrictive than the general model, still captures many real-life applications.
    In particular, it models binary classification problems where one label is more desirable than the other for all agents.
    Common examples include hiring decisions, university admissions, etc.
    Moreover, such decisions are generally correlated over types (e.g., universities cannot admit too many students) --- this is captured by the principal's submodular cost function.
    \item The proof of Lemma~\ref{lem:submodular-binary} can be alternatively interpreted in the following way.
    Without loss of generality, any optimal mechanism corresponds to a point on the convex envelope of the principal's cost function.
    And for submodular cost functions particularly, this convex envelope happens to coincide with the Lov\'{a}sz extension (see \cite{grotschel1984geometric}), which can be derandomized into deterministic mechanisms preserving truthfulness.
    We will further develop this intuition in the next result, which is an efficient algorithm for finding optimal truthful mechanisms for any submodular function.
\end{itemize}

\subsection{An Efficient Algorithm for Computing Optimal Randomized Mechanisms}

In this subsection, we present an algorithm for finding an optimal mechanism with arbitrary submodular cost functions.

Our algorithm, Algorithm~\ref{alg:submodular-randomized}, again builds on the intuition that for optimal mechanisms, only the convex envelope of the cost function matters.
The problem of finding optimal mechanisms can therefore be formulated as a convex program.
However, unlike in Algorithm~\ref{alg:additive-randomized}, with submodular cost functions, it is not clear how one can efficiently evaluate the convex envelope of the cost function.

To get around this issue, instead of parametrizing by the target utilities, we parametrize the convex envelope by the marginal probabilities $\{p_{i, j}\}_{i \in \Theta, j \in [m]}$, where $p_{i, j}$ is the probability that type $i$ gets outcome $o_j$.
One of the key ingredients of Algorithm~\ref{alg:submodular-randomized} is a subroutine (Algorithm~\ref{alg:envelope}) which efficiently interprets each point on the convex envelope as a convex combination of integral points, corresponding to a distribution over combinations of outcomes.
In other words, given the desired marginal probabilities, Algorithm~\ref{alg:envelope} finds a randomized truthful mechanism realizing these marginal probabilities which minimizes the principal's expected cost.

\newcommand{\prob}{\mathrm{prob}}

\begin{algorithm}[t]
\SetAlgoNoLine
\KwIn{The set of types $\Theta$, the principal's submodular cost function $c$, the set of outcomes $O$ (which encodes the common utility function), the reporting structure $R$, and a precision parameter $\varepsilon$.}
\KwOut{An optimal truthful mechanism $M \in \Delta(\Theta \to \mathcal{O})$.}

Compute an $\varepsilon$-approximately optimal solution $p^*_{i,j}$ to the following convex program using the ellipsoid method (see, e.g., \cite{bubeck2015convex}), and call Algorithm~\ref{alg:envelope} with parameters $\{p_{i, j}\}$ for evaluating $\prob(\cdot \mid \{p_{i, j}\})$
\[
\begin{aligned}
    \min \quad & \sum_{O \in \mathcal{O}^\Theta} \prob(O \mid \{p_{i, j}\}_{i \in \Theta, j \in [m]}) \cdot c(O) & \\
    \text{s.t.} \quad & \sum_{j \in [m]} p_{i_1, j} \cdot o_j \ge \sum_{j \in [m]} p_{i_2, j} \cdot o_j & \forall (i_1, i_2) \in R \\
    & \sum_{j \in [m]} p_{i, j} = 1 & \forall i \in \Theta \\
    & p_{i, j} \ge 0 & \forall i \in \Theta, j \in [m]\text{\;}
\end{aligned}
\]

\For{each $O \in \mathcal{O}^\Theta$}{
    $\Pr[M = O] \leftarrow \prob(O \mid \{p_{i, j}^*\})$\; %, where $\{p_{i, j}^*\}$ is an $\varepsilon$-approximately optimal solution to the above\;
}
\Return $M$\;

\caption{Finding an optimal mechanism with submodular cost functions.}
\label{alg:submodular-randomized}
\end{algorithm}

\begin{algorithm}[t]
\SetAlgoNoLine
\KwIn{Marginal probabilities $\{p_{i, j}\}_{i, j}$, where $p_{i, j}$ is the desired probability that type $i$ gets outcome $o_j$.}
\KwOut{A distribution $\{\prob(O \mid \{p_{i, j}\})\}_{O \in \mathcal{O}^\Theta}$ over combinations of outcomes.}

Let $\prob(O) \leftarrow 0$ for all $O \in \mathcal{O}^\Theta$\;
\While{there is some $p_{i, j} > 0$}{
    \For{each $i \in \Theta$}{
        $t_i \leftarrow \max\{j \in [m] \mid p_{i, j} > 0\}$\;
    }
    $\delta \leftarrow \min_i p_{i, t_i}$; \quad $\prob((t_i)_{i \in [n]}) \leftarrow \delta$\;
    \For{each $i \in \Theta$}{
        $p_{i, t_i} \leftarrow p_{i, t_i} - \delta$\;
    }
    return $\{\prob(O)\}_{O \in \mathcal{O}^\Theta}$\;
}

\caption{Algorithm for interpreting the convex envelope.}
\label{alg:envelope}
\end{algorithm}

\begin{theorem}
\label{thm:submodular-randomized}
    When the cost function $c$ is submodular and bounded, for any desired additive error $\varepsilon > 0$, Algorithm~\ref{alg:submodular-randomized} finds an $\varepsilon$-approximately optimal (possibly randomized) truthful mechanism~\footnote{An $\varepsilon$-approximately optimal truthful mechanism is a truthful mechanism whose expected cost is at most $\varepsilon$ larger than the minimum possible cost of any truthful mechanism.} in time $\mathrm{poly}(n, m, \log(1 / \varepsilon))$, even if partial verification is allowed.
\end{theorem}

\begin{proof}%[Proof of Theorem~\ref{thm:submodular-randomized}]
    Consider the program in Algorithm~\ref{alg:submodular-randomized}.
    Observe there are $mn$ variables, namely $\{p_{i, j}\}$, and $O(n^2 + mn)$ linear constraints in the program.
    In order for the program to be efficiently solvable, we only need to show the following three claims hold.
    \begin{itemize}
        \item The subroutine for evaluating the convex envelope, Algorithm~\ref{alg:envelope}, runs in polynomial time.
        \item The objective is convex in $\{p_{i, j}\}$. % (and Lipschitz, which is automatically satisfied for any bounded or normalized cost function $c$).
        \item A first-order oracle, which computes a (sub)gradient of the objective function at any point, can be efficiently implemented.
        (The ellipsoid method also requires an efficient separation oracle, which for our program exists straightforwardly, since there are only $\mathrm{poly}(n, m)$ linear constraints.)
    \end{itemize}
    The second claim and the third claim also imply the correctness of the Algorithm.
    Below we prove the three claims.

    Consider the first claim.
    We only need to show that the while-loop at line~2 repeats only polynomially many times.
    Consider the following potential function $\phi$.
    \[
        \phi(\{p_{i, j}\}) = \sum_{i \in \Theta} \max \{j \in [m] \mid p_{i, j} > 0\}.
    \]
    Before line~2, the value of $\phi$ is at most $mn$.
    Observe that $\phi$ is monotone in $\{p_{i, j}\}$, and the latter never increase during the execution of the loop.
    Moreover, in each repetition of the loop, $\phi$ decreases at least by $1$.
    This is because after the update in line~9, for some $i \in \Theta$, $p_{i, t_i}$ becomes $0$, and as a result, $\max \{j \in [m] \mid p_{i, j} > 0\}$ decreases at least by $1$.
    When $\phi$ becomes $0$, it must be the case that $p_{i, j} = 0$ for any $i \in \Theta, j \in [m]$, so the loop terminates.
    Therefore the while-loop repeats at most $mn$ times, which implies the first claim.

    Now consider the second claim.
    We show that in Algorithm~\ref{alg:envelope}, the output distribution $\{p(O)\}$ minimizes the expected cost
    \[
        \sum_{O \in \mathcal{O}^\Theta} p(O) \cdot c(O)
    \]
    among all distributions whose marginals are $\{p_{i, j}\}$ --- this is equivalent to the second claim.
    We first prove the following characterization of the output distribution $\{p(O)\}$.

    \begin{lemma}
    \label{lem:envelope}
        The output distribution $\{p(O)\}$ of Algorithm~\ref{alg:envelope} is the only distribution over $O^\Theta$ satisfying the following properties.
        \begin{itemize}
            \item $\{p(O)\}$ induce the input marginal probabilities $\{p_{i, j}\}$ over type-outcome pairs.
            \item For any $O_1, O_2 \in \mathcal{O}^\Theta$, if $O_1 \wedge O_2 \notin \{O_1, O_2\}$, then either $p(O_1) = 0$ or $p(O_2) = 0$.
            In other words, no two combinations of outcomes in the support of $\{p(O)\}$ cross.
        \end{itemize}
    \end{lemma}
    \begin{proof}
        The first bullet point is clear from the construction of $\{p(O)\}$.
        We therefore focus on the second bullet point.
        We first show that for any marginal probabilities $\{p_{i, j}\}$ and distribution $\{p'(O)\}$, if (1) $\{p'(O)\}$ induce $\{p_{i, j}\}$ and no two combinations of outcomes in the support of $\{p'(O)\}$ cross, then the topmost combination of outcomes $O_t$, where
        \[
            O_t^i = \max\{o_j \mid j \in [m], p_{i, j} > 0\},
        \]
        must have probability exactly
        \[
            p'(O_t) = \min_{i \in \Theta} p_{i, t_i},
        \]
        where $t_i = \max\{j \in [m] \mid p_{i, j} > 0\}$.

        Let $\delta = \min_{i \in \Theta} p_{i, t_i}$.
        Observe that $p'(O_t) \le \delta$.
        Suppose toward a contradiction that $p'(O_t) < \delta$.
        Since no two combinations of outcomes in the support of $\{p'(O)\}$ cross, we can order the support of $\{p'(O)\}$ as $O_1, \dots, O_\ell$, where $\ell$ is the size of the support, such that for any $i \in [\ell - 1]$,
        \[
            O_i \wedge O_{i + 1} = O_{i + 1}.
        \]
        Clearly we have $O_t \wedge O_1 = O_1$.
        Consider the following two cases.
        \begin{itemize}
            \item $O_t \ne O_1$.
            In other words, there is some $i^* \in \Theta$, such that $O_t^{i^*} > O_1^{i^*}$.
            As a result, for any $k \in [\ell]$,
            \[
                O_k^{i^*} \le O_1^{i^*} < O_t^{i^*} \le o_{t_{i^*}}.
            \]
            Then we have
            \[
                0 < \delta \le p_{i^*, t_{i^*}} = \sum_{k \in [\ell]} p'(O_k) \cdot \mathbb{I}[O_k^{i^*} = o_{t_{i^*}}] = 0,
            \]
            a contradiction.
            \item $O_t = O_1$.
            There is some $i^* \in \Theta$, such that $O_1^{i^*} > O_2^{i^*}$.
            As a result, for any $2 \le k \le \ell$,
            \[
                O_k^{i^*} \le O_2^{i^*} < O_1^{i^*} = o_{t_{i^*}}.
            \]
            So we have
            \[
                \delta \le p_{i^*, t_{i^*}} = \sum_{k \in [\ell]} p'(O_k) \cdot \mathbb{I}[O_k^{i^*} = o_{t_{i^*}}] \le p'(O_1) < \delta,
            \]
            a contradiction.
        \end{itemize}
        So in any case, we must have $p'(O_t) = \delta$.

        Now observe that the above argument does not depend on the fact that for any $i \in \Theta$, $\sum_{j \in [m]} p_{i, j} = 1$.
        Therefore, we can repeatedly apply the characterization of the probability of the topmost combination.
        That is, we first compute the probability of the topmost combination, and subtract the marginal probabilities contributed by this topmost combination from $\{p_{i, j}\}$.
        For the new marginal probabilities, the characterization still applies to the new topmost combination, by which we can determine the probability of that combination.
        This is precisely the procedure implemented in Algorithm~\ref{alg:envelope}.
        By repeatedly applying the characterization, we obtain the unique distribution over $\mathcal{O}^\Theta$ satisfying the conditions of the lemma, which is the output distribution $\{p(O)\}$ of Algorithm~\ref{alg:envelope}.
    \end{proof}

    Given Lemma~\ref{lem:envelope}, we then consider any distribution $\{p'(O)\}$ which (1) induces marginal probabilities $\{p_{i, j}\}$, (2) minimizes the expected cost, and (3) among all distributions satisfying (1) and (2), maximizes the potential function $w$, defined as
    \[
        w(p') = \sum_{O \in \mathcal{O}^\Theta} p'(O) \cdot \left(\sum_{i \in \Theta} O^i\right)^2.
    \]
    The goal is to show such a distribution $\{p'(O)\}$ satisfies the conditions of Lemma~\ref{lem:envelope}, and therefore coincides with $\{p(O)\}$, the output of Algorithm~\ref{alg:envelope}.
    Moreover, such a distribution has the additional property, that it minimizes the expected cost.
    In other words, the output distribution of Algorithm~\ref{alg:envelope} minimizes the expected cost, which is equivalent to the second claim at the beginning of the proof.

    To achieve the above goal, we only need to show that the distribution $\{p'(O)\}$ chosen above has the property, that no two combinations of outcomes in the support of $\{p'(O)\}$ cross.
    Suppose otherwise, i.e., there exist $O_1, O_2 \in \mathcal{O}^\Theta$, such that $O_1 \wedge O_2 \notin \{O_1, O_2\}$ and $p'(O_1) > p'(O_2) > 0$.
    Let $q = p'(O_2)$.
    We show that subtracting $q$ from $p'(O_1)$ and $p'(O_2)$ simultaneously and adding $q$ to $p'(O_1 \wedge O_2)$ and $p'(O_1 \vee O_2)$ simultaneously (1) preserves the marginal probabilities, (2) does not increase the expected cost of $\{p'(O)\}$, and (3) strictly increases the potential function $w$, thus leading to a contradiction.
    (1) clearly holds.
    (2) follows from the submodularity of $c$, i.e.,
    \[
        q \cdot (c(O_1) + c(O_2)) \ge q \cdot (c(O_1 \wedge O_2) + c(O_1 \vee O_2)).
    \]
    And finally, (3) follows from elementary calculation, i.e., whenever $O_1$ and $O_2$ cross,
    \[
        \left(\sum_{i \in \Theta} O_1^i\right)^2 + \left(\sum_{i \in \Theta} O_2^i\right)^2 < \left(\sum_{i \in \Theta} \min(O_1^i, O_2^i)\right)^2 + \left(\sum_{i \in \Theta} \max(O_1^i, O_2^i)\right)^2.
    \]
    This establishes second claim.
    
    As for the third claim, i.e., the existence of an efficient first-order oracle, observe that the distribution $\{\prob(O \mid \{p_{i, j}\})\}$ output by Algorithm~\ref{alg:envelope} is piecewise linear in $\{p_{i, j}\}$.
    On the other hand, the objective function is linear in the output distribution, and is therefore piecewise linear in $\{p_{i, j}\}$.
    This implies that (sub)gradients of the objective function can be easily computed, and concludes the proof of the theorem.
\end{proof}

We make a few remarks regarding Algorithm~\ref{alg:submodular-randomized}.
\begin{itemize}
    \item In addition to the ellipsoid method, one may also apply gradient-based methods, e.g., projected gradient descent, to solve the convex program in Algorithm~\ref{alg:submodular-randomized}.
    Gradient-based methods generally perform better in practice, and they usually have better dependence on $m$ and $n$ but worse (polynomial) dependence on $1 / \varepsilon$.
    \item Observe that Algorithm~\ref{alg:submodular-randomized} outputs a randomized mechanism such that, in its support, no two combinations of outcomes cross.
    Therefore, when restricted to binary outcomes ($m = 2$), Lemma~\ref{lem:submodular-binary} gives a way to round the randomized mechanism output by Algorithm~\ref{alg:submodular-randomized} into a deterministic one.
    This gives an alternative way (in addition to Corollary~\ref{cor:submodular-deterministic}) of computing optimal deterministic mechanisms restricted to binary outcomes.
\end{itemize}

\section{Omitted Definitions and Remarks in Section~\ref{sec:additive}}

\subsection{Single-Peaked Preferences}

Below we give a definition of single-peaked preferences.
\begin{definition}
    Let $O$ be the space of outcomes, $\Theta$ the space of types, and for each $i \in \Theta$, $u_i: O \to \mathbb{R}_+$ the utility function of an agent of type $i$.
    $\{u_i\}_i$ are single-peaked if there exists an ordering $\prec$ over $O$, such that for each type $i$ the following holds: there exists a most preferred outcome $o^i$.
    Moreover, for any two outcomes $o_{j_1} \prec o_{j_2}$,
    \begin{itemize}
        \item if $o_{j_2} \prec o^i$, then $u_i(o_{j_1}) \le u_i(o_{j_2}) \le u_i(o^i)$;
        \item if $o^* \prec o_{j_1}$, then $u_i(o^i) \ge u_i(o_{j_1}) \ge u_i(o_{j_2})$.
    \end{itemize}
\end{definition}
In words, the above definition says that the outcomes can be ordered in a line, such that for each type, there exists a most preferred outcome.
Moreover, on both sides of this most preferred outcome, the closer an outcome is to the most preferred the outcome, the higher the utility is for that outcome.

\begin{comment}
\subsection{Illustration of Algorithm~\ref{alg:additive-deterministic}}

See Figure~\ref{fig:flow} for an illustration of Algorithm~\ref{alg:additive-deterministic}.

\begin{figure}[ht]
\centering
\includegraphics[width=0.48\linewidth]{flow}
\includegraphics[width=0.48\linewidth]{cut}
\caption{
    An example of the graph constructed in Algorithm~\ref{alg:additive-deterministic}.
    As highlighted in the left graph, each row corresponds to an outcome and each column corresponds to a type.
    The horizontal edges with infinite capacity correspond to the fact that type $2$ can misreport as type $1$.
    The right graph gives a possible $s$-$t$ min-cut, which corresponds to a mechanism where $M(1) = o_2$, $M(2) = (o_3)$, and $M(3) = o_3$.
    The horizontal edges make sure that type $1$ never gets a more desirable outcome than type $2$, so type $2$ never misreports.
    The cost $c(M)$ of the mechanism $M$ is equal to the value of the min-cut, which is $c_1(o_2) + c_2(o_3) + c_3(o_3)$.
}
\label{fig:flow}
\end{figure}
\end{comment}

\subsection{Remarks on Algorithm~\ref{alg:additive-deterministic}}

We make two remarks regarding Algorithm~\ref{alg:additive-deterministic}.
\begin{itemize}
    \item For finding an optimal deterministic mechanism, the precise values of the agents' utility functions do not matter.
    Consequently, Algorithm~\ref{alg:additive-deterministic} works as long as all types {\em order} the outcomes in the same way.
    \item With minor modifications, Algorithm~\ref{alg:additive-deterministic} can handle {\em costly} misreporting, in which there is a fixed (non-negative) cost for type $i$ to report as type $i'$.
    Partial verification is a special case of costly misreporting: reporting either costs the agent $0$ or $\infty$, and the reporting structure $R$ is the set of all reporting actions which cost $0$.
    The key modification which allows Algorithm~\ref{alg:additive-deterministic} to handle costly misreporting is that the edges used to model the reporting structure can be diagonal (as opposed to horizontal), where the slope of the edge depends on each type's utility function and the cost of misreporting.
    We will not expand on this in the current paper.
\end{itemize}

\subsection{Remarks on Lemma~\ref{lem:additive-convex}}

We make a few remarks regarding Lemma~\ref{lem:additive-convex}.
\begin{itemize}
    \item The proof we present is a combination of several concrete arguments.
    There is an alternative relatively high-level, and sometimes more useful, interpretation of the lemma, which is based on a convex program formulation of the problem.
    We will make heavy use of this alternative interpretation in the rest of the paper, especially when dealing with randomized mechanisms.
    \item Throughout the paper we assume payments are not allowed.
    One may show that with payments, there always exists an optimal truthful mechanism that is deterministic, as long as both agents and the principal value payments linearly.
    Moreover, there exist relatively simple algorithms for computing an optimal mechanism with payments.
    We will not expand on this in the current paper.
\end{itemize}

\subsection{Remarks on Algorithm~\ref{alg:additive-randomized}}

We make a few remarks regarding Algorithm~\ref{alg:additive-randomized}.
\begin{itemize}
    \item Algorithm~\ref{alg:additive-randomized} gives a constructive proof that finding an optimal truthful mechanism is always no harder than finding an optimal truthful deterministic mechanism with convex costs.
    As a result, a faster algorithm for the latter problem would imply a faster algorithm for the former.
    \item As a byproduct, Algorithm~\ref{alg:additive-randomized} shows that in general, to achieve the minimum cost, it suffices to randomize only between two outcomes for each type, .
\end{itemize}

\section{Omitted Proofs in Section~\ref{sec:additive}}
\label{app:additive}

\begin{proof}[Proof of Theorem~\ref{thm:np-hardness-without-revelation}]
    We give a reduction from $\mathsf{MinSAT}$.
    Fix a $\mathsf{MinSAT}$ instance with $n$ variables, $\{x_i\}_{i \in [n]}$, and $m$ clauses, $\{C_j\}_{j \in [m]}$, and let $\ell_{j, k} \in C_j$ be the $k$-th literal in clause $C_j$.
    We construct an $\mathsf{AMD}$ instance as follows.
    \begin{itemize}
        \item Create a type for each variable, each literal, and each clause, i.e., $\Theta = \{x_i, x_i^+, x_i^-\}_{i \in [n]} \cup \{C_j\}_{j \in [m]}$.
        \item There are two possible outcomes, $\mathcal{O} = \{o^+, o^-\}$.
        Moreover, for any type $\theta \in \Theta$, $u_\theta(o^+) = 1$ and $u_\theta(o^-) = 0$.
        \item The principal's cost is as follows.
        \begin{itemize}
            \item For each literal $\ell$, $c_\ell(o^+) = c_\ell(o^-) = 0$.
            \item For each variable $x_i$, $c_{x_i}(o^+) = 0$ and $c_{x_i}(o^-) = m + 1$, so any optimal mechanism never assigns $o^-$ to a variable.
            \item For each clause $C_j$, $c_{C_j}(o^+) = 1$ and $c_{C_j}(o^-) = 0$, so any optimal mechanism minimizes the number of clauses which get outcome $o^+$.
        \end{itemize}
        \item The reporting structure $R$ is as follows.
        \begin{itemize}
            \item Each literal $\ell$ can only report itself.
            \item Each variable $x_i$ can report itself and its two literals $x_i^+$ and $x_i^-$.
            \item Each clause $C_j$ can report itself, all variables, and all literals $\ell_{j, k} \in C_j$ contained in $C_j$.
        \end{itemize}
    \end{itemize}
    
    Now consider the structure of optimal solutions for the above $\mathsf{AMD}$ instance.
    First observe that without loss of generality, any optimal solution assigns $o^+$ only to types which report literals.
    Moreover, for each variable $x_i$, any optimal solution assigns $o^+$ to exactly one of $x_i^+$ and $x_i^-$.
    So the problem boils down to choosing between the two literals for each variable.

    On the other hand, each clause $C_j$ will report any literal that is contained in $C_j$ and assigned outcome $o^+$, as long as possible.
    Whenever this happens, the principal incurs cost $1$ from this clause.
    In other words, the principal incurs cost $1$ from a clause iff one of the literals contained in the clause is assigned outcome $o^+$, i.e., iff the clause is satisfied.
    The total cost of the mechanism is $k$, where $k$ is the number of clauses satisfied.
    This encodes precisely the $\mathsf{MinSAT}$ instance.
\end{proof}

\begin{proof}[Proof of Theorem~\ref{thm:np-hardness-with-general-utility}]
    Consider the following reduction from $\mathsf{MinSAT}$.
    Fix a $\mathsf{MinSAT}$ instance with $n$ variables, $\{x_i\}_{i \in [n]}$, and $m$ clauses, $\{C_j\}_{j \in [m]}$, and let $\ell_{j, k} \in C_j$ be the $k$-th literal in clause $C_j$.
    We construct an $\mathsf{AMD}$ instance as follows.
    \begin{itemize}
        \item Create a type for each variable, each literal, and each clause, i.e., $\Theta = \{x_i, x_i^+, x_i^-\}_{i \in [n]} \cup \{C_j\}_{j \in [m]}$.
        \item There are three possible outcomes, $\mathcal{O} = \{o^+, o^-, o^0\}$.
        \item Let $N_1 >> N_2 > m$ be large (but polynomial in $n$ and $m$) numbers.
        The principal's cost is as follows.
        \begin{itemize}
            \item For each variable $x_i$, $c_{x_i}(o^+) = c_{x_i}(o^-) = 0$, and $c_{x_i}(o^0) = N_1$.
            As a result, an optimal mechanism never assigns $o^0$ to a variable.
            \item For each positive literal $\ell^+$, $c_{\ell^+}(o^+) = N_2$, $c_{\ell^+}(o^0) = 0$, and $c_{\ell^+}(o^-) = N_1$.
            For each negative literal $\ell^-$, $c_{\ell^-}(o^+) = N_1$, $c_{\ell^-}(o^0) = 0$, and $c_{\ell^-}(o^-) = N_2$.
            We will see later that for any variable $x_i$, an optimal mechanism assigns precisely one of its literals the outcome with cost $N_2$, and the other outcome $o^0$ with cost $0$.
            \item For each clause $C_j$, $c_{C_j}(o^+) = c_{C_j}(o^-) = 0$, and $c_{C_j}(o^0) = 1$.
        \end{itemize}
        \item The types' utility functions are as follows.
        \begin{itemize}
            \item For each variable $x_i$, $u_{x_i}(o^+) = u_{x_i}(o^-) = u_{x_i}(o^0)$.
            Note that the numerical values of the utility functions do not matter for deterministic mechanisms.
            \item For each positive literal $\ell^+$, $u_{\ell^+}(o^+) > u_{\ell^+}(o^0) > u_{\ell^+}(o^-)$.
            For each negative literal $\ell^-$, $u_{\ell^-}(o^+) < u_{\ell^-}(o^0) < u_{\ell^-}(o^-)$.
            \item For each clause $C_j$, $u_{C_j}(o^0) > u_{C_j}(o^+) = u_{C_j}(o^-)$.
        \end{itemize}
        \item The reporting structure $R$ is as follows.
        \begin{itemize}
            \item Each variable $x_i$ can only report itself.
            \item Each literal $\ell$ can report itself or the variable it corresponds to.
            \item Each clause $C_j$ can report itself, any literal $\ell \in C_j$ contained in the clause, or the variable $\ell$ corresponds to.
        \end{itemize}
    \end{itemize}

    Now consider the structure of optimal deterministic mechanisms.
    For each variable $x_i$, an optimal mechanism assigns either $o^+$ or $o^-$.
    Moreover, for $x_i$'s two literals, if $x_i$ is assigned $o^+$ (resp.\ $o^-$), then the mechanism always assigns $x_i^+$ $o^+$ (resp.\ $o^0$) and $x_i^-$ $o^0$ (resp.\ $o^-$).
    One may check this is the only way to minimize cost subject to incentive compatibility.
    So conceptually, the mechanism chooses exactly one value for each variable, where assigning $o^0$ to $x_i^+$ (resp.\ $x_i^-$) corresponds to choosing value $1$ (resp.\ $0$) for $x_i$.

    For each clause $C_j$, if any of the literals contained in $C_j$ is chosen (i.e., is assigned outcome $o^0$), then to prevent $C_j$ from misreporting that literal, the mechanism must assign $C_j$ outcome $o^0$, at a cost of $1$.
    This corresponds to the case where the clause is satisfied.
    Otherwise, if none of the literals in $C_j$ is chosen, the mechanism assigns either $o^+$ or $o^-$ to $C_j$, at a cost of $0$.
    The total cost of the mechanism is then $n N_2 + k$, where $k$ is the number of clauses satisfied.
    This encodes precisely the $\mathsf{MinSAT}$ instance.
\end{proof}

\begin{proof}[Proof of Theorem~\ref{thm:additive-deterministic}]
    First consider the runtime of Algorithm~\ref{alg:additive-deterministic}.
    The bottleneck is finding an $s$-$t$ min-cut on the graph $G$, which has $mn + 2$ vertices and at most $mn^2 + (m + 1)n$ edges.
    Therefore, it is sufficient to show that one can replace the infinite capacities with capacity $W + 1$.
    
    We first prove that any horizontal edge with capacity $W + 1$ does not belong to any min-cut.
    Suppose in some min-cut, for some $j \in [m]$, a horizontal edge from $(i_1, o_j) \in S$ to $(i_2, o_j) \notin S$ is cut.
    We argue that including all out-neighbors of $(i_1, o_j)$ through horizontal edges into $S$ strictly decreases the capacity of the cut.
    For each of these horizontal out-neighbors, by including it in $S$, we decrease the cut value by $W + 1$ (from one horizontal edge), and possibly incur an additional cost from the edge between that neighbor and its vertical out-neighbor, whose capacity is at most $W$.
    Because we take the transitive closure of $R$, the newly included vertices do not have any horizontal out-neighbor out of $S$, so the total cost decreases at least by $1$.
    A similar argument shows that edges leaving $S$ can be replaced to have capacity $W+1$ as well.

    Now we move on to proving the correctness of Algorithm~\ref{alg:additive-deterministic}.
    We assume the infinite-capacity edges still have capacity $\infty$ (rather than $W+1$), which simplifies our argument.
    Observe that with infinite capacities, taking the transitive closure of $R$ in Line~3-5 of Algorithm~\ref{alg:additive-deterministic} makes no difference.
    We prove the correctness for the algorithm without this step.

    The argument consists of two parts.
    First we show there is a one-to-one correspondence between all finite-capacity \emph{downward-closed} $s$-$t$ cuts and all deterministic truthful mechanisms, where the capacity of the cut is the same as the cost of the mechanism.
    We then show that taking the \emph{downward closure} of any cut does not increase its capacity, and as a result, we only need to consider downward-closed cuts.
    These two claims together imply the correctness of Algorithm~\ref{alg:additive-deterministic}.

    Formally, a cut $(S, \overline{S})$ is downward closed, if for any $i \in \Theta$ and $1 \le j_1 < j_2 \le m$,
    \[
        (i, o_{j_2}) \in S \; \Longrightarrow \; (i, o_{j_1}) \in S.
    \]
    Fix a downward closed cut $(S, \overline{S})$, we construct a mechanism $M: \Theta \to \mathcal{O}$ in the same way as in Line~17-19 of Algorithm~\ref{alg:additive-deterministic}.
    That is, for all $i \in \Theta$,
    \[
        M(i) = \max\{o_{j'} \in \mathcal{O} \mid (i, o_{j'}) \in S\}.
    \]
    The one-to-one correspondence follows immediately from the definition of $M$.
    We now argue $(S, \overline{S})$ has finite capacity iff $M$ is truthful.
    Notice that $(S, \overline{S})$ has finite capacity iff no horizontal edge is cut, i.e., iff $M(i_1) \ge M(i_2)$ for all $(i_1, i_2) \in R$, which is precisely the condition for the truthfulness of $M$.
    Moreover, whenever $(S, \overline{S})$ has finite capacity, the capacity is equal to the cost of the truthful mechanism $M$.

    Now we prove the second claim, i.e., taking the downward closure does not increase the capacity of the cut.
    We first define the downward closure.
    Given any $s$-$t$ cut $(S, \overline{S})$, the downward closure $(C(S), \overline{C(S)})$ is defined such that for all $i \in \Theta$,
    \[
        (i, o_j) \in C(S) \iff j \le \max\{j' \in [m] \mid (i, o_{j'}) \in S\}.
    \]
    We show below that the capacity of $C(S)$ is no larger than that of $S$.

    If some horizontal edge is cut in $S$, then the cut has capacity $\infty$ and the claim is trivial.
    Suppose no horizontal edge is cut in $S$.
    Because the set of vertical edges cut in $C(S)$ is a subset of those cut in $S$, we only need to show that no horizontal edge is cut in $C(S)$.
    Suppose for contradiction that some horizontal edges are cut in $C(S)$.
    Let $i_1, i_2 \in \Theta$ and $j \in [m]$ be such that $e = ((i_1, o_j), (i_2, o_j))$ is one of the highest horizontal edges being cut in $C(S)$.
    By the choice of $e$, it must be the case that $(i_1, o_j) \in S$, and since $S \subseteq C(S)$, we have $(i_2, o_j) \notin S$.
    Therefore, the same edge $e$ is also cut in $S$, which leads to a contradiction.
\end{proof}

\begin{proof}[Proof of Lemma~\ref{lem:additive-convex}]
    We prove the lemma by construction.
    Let $M$ be any (possibly randomized) optimal truthful mechanism.
    We construct a deterministic truthful mechanism from $M$ whose cost is no larger than that of $M$.

    First we show it suffices for $M$ to randomize between only two consecutive outcomes for each type $i$.
    Let $p_i(o_j)$ be the probability that type $i$ receives outcome $o_j$.
    Suppose for some type $i$, there exist $j_1, j_2 \in [m]$, where $j_2 - j_1 > 1$, $p_i(o_{j_1}) > 0$, and $p_i(o_{j_2}) > 0$.
    We argue that one can move probability mass from $o_{j_1}$ and $o_{j_2}$ to $o_{j_3}$, where $j_3 = j_1 + 1$ lies between $j_1$ and $j_2$, without violating truthfulness or increasing the total cost.

    Let $0 < \alpha < 1$ be such that $o_{j_3} = \alpha o_{j_1} + (1 - \alpha) o_{j_2}$.
    Without loss of generality suppose $p_i(o_{j_1}) / \alpha \le p_i(o_{j_2}) / (1 - \alpha)$.
    For brevity let $p = p_i(o_{j_1})$.
    We show that the following operation achieves the above goal.
    \begin{itemize}
        \item Decrease $p_i(o_{j_1})$ by $p$.
        \item Decrease $p_i(o_{j_2})$ by $(1 - \alpha) \cdot p / \alpha \le p_i(o_{j_2})$.
        \item Increase $p_i(o_{j_3})$ by $p / \alpha$.
    \end{itemize}
    Observe that (1) after the operation, the probabilities of each outcome still sum to $1$, and (2) type $i$ receives exactly the same expected utility.
    The principal's cost changes by
    \begin{align*}
        & \quad \frac{p}{\alpha} \cdot c_i(o_{j_3}) - p \cdot c_i(o_{j_1}) - \frac{(1 - \alpha) \cdot p}{\alpha} \cdot c_i(o_{j_2}) \\
        & = \frac{p}{\alpha} \cdot \left(c_i(o_{j_3}) - \alpha \cdot c_i(o_{j_1}) - (1 - \alpha) \cdot c_i(o_{j_2}) \right) \\
        & = \frac{p}{\alpha} \cdot \left(c_i^\ell(o_{j_3}) - \alpha \cdot c_i^\ell(o_{j_1}) - (1 - \alpha) \cdot c_i^\ell(o_{j_2}) \right) \tag{$c_i^\ell$ extends $c_i$} \\
        & \le \frac{p}{\alpha} \cdot \left(c_i^\ell(o_{j_3}) - c_i^\ell(\alpha \cdot o_{j_1} + (1 - \alpha) \cdot o_{j_2}) \right) \tag{convexity of $c_i^\ell$} \\
        & = \frac{p}{\alpha} \cdot \left(c_i^\ell(o_{j_3}) - c_i^\ell(o_{j_3}) \right) = 0. \tag{$o_{j_3} = \alpha \cdot o_{j_1} + (1 - \alpha) \cdot o_{j_2}$}
    \end{align*}
    In other words, the total cost does not increase.

    We then apply the above operation in the following way.
    Fix $i$, and let $j_- = \min \{j \mid p_i(o_j) > 0\}$, and $j_+ = \max \{j \mid p_i(o_j) > 0\}$.
    As long as $j_+ - j_- > 1$, apply the operation to $i$, $j_-$ and $j_+$.
    Observe that each time we apply the operation, $j_+ - j_-$ decreases by at least $1$, so eventually we must stop and $j_+ - j_- \le 1$.
    Performing this for each $i$ yields a mechanism which randomizes only between two consecutive outcomes for each type, without increasing the total cost.
    Without loss of generality, from now on, we assume $M$ has this property.

    Now we show there is a way to round $M$, producing a distribution over deterministic truthful mechanisms, such that the expected cost of this distribution is precisely the cost of $M$.
    As a result, there exists one mechanism in the support of the distribution, whose cost is upper bounded by that of $M$, which is our desired deterministic truthful mechanism.
    
    For each type $i$, let $j_i = \min \{j \mid p_i(j) > 0\}$ and $\alpha_i = p_i(j_i + 1)$.
    Note that $0 \le \alpha_i < 1$ for all $i$.
    For any $r \in [0, 1]$, let $M_r$ be the deterministic mechanism defined such that for each type $i$,
    \[
        M_r(i) =
        \begin{cases}
            o_{j_i + 1}, & r \le \alpha_i, \\
            o_{j_i}, & \text{otherwise}.
        \end{cases}.
    \]
    We first argue that $M_r$ is truthful for any $r \in [0, 1]$.
    Fix any pair $(i_1, i_2) \in R$.
    Given that $M$ itself is truthful, we proceed by considering the following two cases.
    \begin{itemize}
        \item $j_{i_1} > j_{i_2}$.
        In such cases, we always have $M_r(i_1) \ge M_r(i_2)$, so $i_1$ has no incentive to report $i_2$.
        \item $j_{i_1} = j_{i_2}$ and $\alpha_{i_1} \ge \alpha_{i_2}$.
        For any $r \in [0, 1]$, we have
        \[
            r \le \alpha_{i_1} \; \Longleftarrow \; r \le \alpha_{i_2}.
        \]
        So again, regardless of $r$, $M_r(i_1) \ge M_r(i_2)$.
    \end{itemize}
    Applying the above argument to each pair $(i_1, i_2) \in R$ establishes the truthfulness of $M_r$ for any $r \in [0, 1]$.

    Now consider the distribution over deterministic mechanisms $M_r$ when $r$ is uniformly distributed over $[0, 1]$.
    We show that the expected cost of $M_r$ is equal to the cost of $M$:
    \begin{align*}
        \mathbb{E}_r[c(M_r)] & = \sum_{i \in \Theta} \left(\Pr[r \le \alpha_i] \cdot c_i(o_{j_i + 1}) + \Pr[r > \alpha_i] \cdot c_i(o_{j_i})\right) \\
        & = \sum_{i \in \Theta} (\alpha_i \cdot c_i(o_{j_i + 1}) + (1 - \alpha_i) \cdot c_i(o_{j_i})) \\
        & = \sum_{i \in \Theta} (p_i(j_i + 1) \cdot c_i(o_{j_i + 1}) + p_i(j_i) \cdot c_i(o_{j_i})) \\
        & = \sum_{i \in \Theta} \mathbb{E}_M[c_i(M(i))] \\
        & = c(M),
    \end{align*}
    which concludes the proof.
\end{proof}

\begin{proof}[Proof of Theorem~\ref{thm:additive-randomize}]
    We prove the correctness first.
    Observe that the problem of finding a (randomized) optimal mechanism can be written as the following linear program.
    \[
    \begin{aligned}
        \min \qquad & \sum_{i \in \Theta = [n], j \in [m]} p_{i, j} \cdot c_i(o_j) & \\
        \text{subject to} \qquad & u_i = \sum_j p_{i, j} \cdot o_j & \forall i \in \Theta \\
        & u_{i_1} \ge u_{i_2} & \forall (i_1, i_2) \in R \\
        & \sum_j p_{i, j} = 1 & \forall i \in \Theta \\
        & p_{i, j} \ge 0 & \forall i \in \Theta, j \in [m].
    \end{aligned}
    \]
    Here, $u_i$ is the expected utility of type $i$, and $p_{i, j}$ is the probability of assigning type $i$ outcome $o_j$.
    This is not the most succinct LP formulation of the problem, but it capture the structure of the problem in a way that is more useful for our analysis.

    Now fix $\{u_i\}_i$ and consider the optimal choice of $\{p_{i, j}\}_{i, j}$.
    This can be solved separately for each type $i$, by considering the following linear program (with the additional constraints that $\{p_{i, j}\}_j$ are nonnegative and sum up to $1$ for all $i$).
    \begin{align*}
        \min\quad & \sum_j p_{i, j} \cdot c_i(o_j) \\
        \text{s.t.}\quad & \sum_j p_{i, j} \cdot o_j = u_i.
    \end{align*}
    This is precisely evaluating the convex envelope $c_i^-$ of $c_i$ at $u_i \in [o_1, o_m]$.
    Consequently, the problem of finding a (randomized) optimal mechanism can be rewritten as the following convex program.
    \[
    \begin{aligned}
        \min\quad & \sum_{i \in \Theta} c_i^-(u_i) & \\
        \text{s.t.}\quad & u_{i_1} \le u_{i_2}, & \forall (i_1, i_2) \in R, \\
        & u_i \in [o_1, o_m], & \forall i \in \Theta.
    \end{aligned}
    \]

    Now observe that the reformulated program cannot distinguish between $\{c_i\}_{i \in \Theta}$ and $\{\widehat{c_i}\}_{i \in \Theta}$, where $\widehat{c_i}$ is $c_i^-$ restricted to $\mathcal{O}$ as in Algorithm~\ref{alg:additive-randomized} --- the two cost functions simply induce exactly the same program.
    Moreover, observe that the newly constructed cost function $\{\widehat{c_i}\}_{i \in \Theta}$ is convex, according to Definition~\ref{def:convex-costs}.
    Given this convexity, Lemma~\ref{lem:additive-convex} implies that there exists a deterministic mechanism for $\{\widehat{c_i}\}_i$ which is optimal.
    In other words, there exists an optimal solution $\{u_i\}_i$ to the reformulated program in which $u_i \in \mathcal{O}$ for each $i \in \Theta$.
    Algorithm~\ref{alg:additive-randomized} finds such a solution $\{u_i\}_i$ by calling Algorithm~\ref{alg:additive-deterministic}.

    Now the only problem left is to recover $\{p_{i, j}\}_j$ from $u_i$ for each type $i$.
    This is done in Line~6 of Algorithm~\ref{alg:additive-randomized}.
    Since the output mechanism $M$ implements $\{u_i\}_i$ in an optimal way, it is a truthful mechanism that minimizes the principal's cost.
    This establishes the correctness of Algorithm~\ref{alg:additive-randomized}.

    Now we consider the time complexity.
    Compared to Algorithm~\ref{alg:additive-deterministic}, the additional steps in Algorithm~\ref{alg:additive-randomized} include (1) computing $\{\widehat{c_i}\}_i$ in Line~2, and (2) interpreting $\widehat{M}(i)$ as an optimal convex combination of outcomes in Line~6.
    We will show that both operations can be done in time $O(mn)$, i.e., linear in the size of the input.
    The time complexity of Algorithm~\ref{alg:additive-randomized} then follows immediately.

    For computing the convex envelope of $c_i$, one may use the classical algorithm by \citet{andrew1979another}, which scans $\mathcal{O}$ from left to right, and maintains a stack containing the partial convex envelope of $c_i$ on $\{o_1, \dots, o_j\}$ for every $j \in [m]$.
    The algorithm runs in time $O(m)$.
    
    Once we know $\{\widehat{c_i}\}$, to find an optimal convex combination for a target utility $u_i$, we first find in time $O(m)$ the largest integer $\ell \in [m]$ such that $o_\ell \le u_i$ and $c_i(o_\ell) = \widehat{c_i}(o_\ell)$, and the smallest integer $r \in [m]$ such that $o_r \ge u_i$ and $c_i(o_r) = \widehat{c_i}(o_r)$.
    If $\ell = r$, then we output $o_\ell = o_r$.
    Otherwise, there is a unique $\alpha \in (0, 1)$ such that randomizing between $o_\ell$ and $o_r$ gives expectation $\alpha \cdot o_\ell + (1 - \alpha) \cdot o_r = u_i$.
    The convex envelope is linear between $o_l$ and $o_r$, and hence $\alpha \cdot c_i(o_\ell) + (1 - \alpha) \cdot c_i(o_r) = \widehat{c_i}(u_i)$.
    Then, we can set $M(i)$ to $o_\ell$ with probability $\alpha$ and to $o_r$ with probability $(1 - \alpha)$.
    Performing the above for every type $i$ takes $O(mn)$ time.
\end{proof}

% \section{Omitted Proofs in Appendix~\ref{app:submodular}}
% \label{app:submodular}

\end{document}